\documentclass[a4paper,twocolumn,11pt,accepted=2023-05-08]{quantumarticle}
\pdfoutput=1
\usepackage[numbers,sort&compress]{natbib}
\usepackage[utf8]{inputenc}
\usepackage[english]{babel}
\usepackage[T1]{fontenc}
\usepackage{amsmath}

\usepackage{tensor}
\usepackage{amsmath}
\usepackage{amsthm}
\usepackage[dvipdfmx]{graphicx}
\usepackage{float}

\usepackage{comment}
\usepackage{amsfonts}
\usepackage{bm}
\usepackage{braket}
\usepackage{xcolor}
\usepackage[colorlinks=true,linkcolor=cyan,urlcolor=teal,citecolor=cyan]{hyperref}

\usepackage{enumerate}
\usepackage{mathtools}
\usepackage{bbm}
\usepackage{mathtools}

\usepackage{multirow}
\usepackage{graphicx}
\usepackage{threeparttable}

\newtheorem{theorem}{Theorem}
\newtheorem*{theorem*}{Theorem}

\newtheorem{proposition}[theorem]{Proposition}
\newtheorem{corollary}[theorem]{Corollary}

\newtheorem{lemma}[theorem]{Lemma}
\newtheorem*{lemma*}{Lemma}

\DeclarePairedDelimiter\ceil{\lceil}{\rceil}

\newcommand{\cov}{\overset{\mathrm{cov}}{\succ}}
\newcommand{\cova}{\cov_{\mathrm{a}}}

\newcommand{\ii}{\mathrm{i}}

\newcommand{\dd}{\mathrm{d}}

\begin{document}

\title{Smooth Metric Adjusted Skew Information Rates}
\author{Koji Yamaguchi}
\affiliation{Department of Applied Mathematics, University of Waterloo, Waterloo, Ontario, N2L
3G1, Canada}
\orcid{0000-0002-9723-6145}
\email{koji.yamaguchi@uwaterloo.ca}
\author{Hiroyasu Tajima}
\affiliation{Department of Communication Engineering and Informatics,
University of Electro-Communications, 1-5-1 Chofugaoka, Chofu, Tokyo, 182-8585, Japan}
\affiliation{JST, PRESTO, 4-1-8 Honcho, Kawaguchi, Saitama, 332-0012, Japan}

\begin{abstract}
Metric adjusted skew information, induced from quantum Fisher information, is a well-known family of resource measures in the resource theory of asymmetry. 
However, its asymptotic rates are not valid asymmetry monotone since it has an asymptotic discontinuity. We here introduce a new class of asymmetry measures with the smoothing technique, which we term smooth metric adjusted skew information. We prove that its asymptotic sup- and inf-rates are valid asymptotic measures in the resource theory of asymmetry. Furthermore, it is proven that the smooth metric adjusted skew information rates provide a lower bound for the coherence cost and an upper bound for the distillable coherence. 
\end{abstract}

\maketitle

\section{Introduction}
Symmetry and conservation laws are basic yet profound concepts in modern physics.
As first pointed out by Wigner \cite{wigner_messung_1952} and further examined by Araki and Yanase \cite{araki_measurement_1960}, conservation laws place limitations on the precision of the measurements. Based on this observation, Wigner and Yanase \cite{wigner_information_1963} introduced a measure of information content in the presence of a conservation law, called the Wigner-Yanase skew information. It quantifies how a state is ``askew" (i.e., non-diagonal, or coherent) to the eigenbasis of the conserved quantity. This quantity was later shown to be closely related to quantum information geometry \cite{gibilisco_characterisation_2001,gibilisco_wigneryanase_2003}. Generalizing this connection, metric adjusted skew information is introduced from a family of quantum Fisher information metrics as quantifiers of information content relative to a conserved quantity \cite{hansen_metric_2008}. 

A deeper understanding of the above relation between symmetries and skew information has been gained through the recent development of quantum resource theories. 
Quantum resource theories \cite{chitambar_quantum_2019} provide a powerful framework for studying the connection between a resource and the physical restriction posed by our inability to prepare a state and perform an operation. Due to its wide applicability, various kinds of resources have been investigated, such as entanglement \cite{horodecki_quantum_2009}, coherence \cite{aberg_quantifying_2006,baumgratz_quantifying_2014,streltsov_colloquium_2017}, athermality \cite{brandao_resource_2013,brandao_second_2015} and asymmetry \cite{gour_resource_2008}. The resource theory of asymmetry is one of the most actively studied quantum resource theories, where the dynamics and states are restricted by symmetries of the system.

In the resource theory of asymmetry, a class of coherence called asymmetry is considered to be a resource to implement operations that break the symmetry. Asymmetry captures consequences of symmetry that cannot be captured by the Noether theorem \cite{marvian_extending_2014} and has various applications. For example, time-translation asymmetry is mandatory for creating accurate clocks \cite{giovannetti_quantum-enhanced_2001,giovannetti_quantum_2006,giacomini_quantum_2019,schnabel_quantum_2010,marvian_coherence_2020,woods_autonomous_2019} and accelerating quantum operations \cite{marvian_quantum_2016}, and is known to be a resource independent of entropy in quantum thermodynamics \cite{lostaglio_description_2015}. Furthermore, asymmetry is shown to be an essential resource in various problems under conservation laws such as quantum measurements \cite{wigner_messung_1952,araki_measurement_1960,yanase_optimal_1961,ozawa_conservation_2002,korzekwa_resource_2003,tajima_coherence-variance_2019,kuramochi_wigner-araki-yanase_2022}, gate implementation in quantum computing \cite{ozawa_conservative_2002,tajima_uncertainty_2018,tajima_coherence_2020,tajima_universal_2021,tajima_universal_2022}, quantum error correction \cite{kubica_using_2021,zhou_new_2021,yang_optimal_2022,tajima_universal_2021,liu_quantum_2022,tajima_universal_2022}, and the Hayden–Preskill model for evaporating black holes \cite{tajima_universal_2021,tajima_universal_2022}.

As a resource under conservation laws, asymmetry is closely related to metric adjusted skew information. Indeed, metric adjusted skew information plays a significant role in both one-shot and asymptotic settings. In the one-shot setting, a family of metric adjusted skew information is known as valid resource measures \cite{zhang_detecting_2017,takagi_skew_2019} for $U(1)$ symmetry. It is lifted to a family of skew information matrices \cite{kudo_fisher_2022} for a general Lie group $G$, which are valid asymmetry measures as well as $G$-asymmetry \cite{vaccaro_tradeoff_2008} and the relative entropy of $G$-frameness \cite{gour_measuring_2009}. The convertibility conditions in the one-shot regime \cite{gour_resource_2008,marvian_mashhad_symmetry_2012,gour_quantum_2018,yamaguchi_beyond_2022,gour_role_2022} are significantly simplified in the asymptotic regime, and a thermodynamic structure based on skew information appears for i.i.d. pure states \cite{gour_resource_2008,marvian_operational_2022}. Furthermore, the conversion theory in the non-i.i.d. regime \cite{yamaguchi_beyond_2022} has been established by extending the quantum Fisher information, which is an example in the family of metric adjusted skew information.

Despite the above progress in the resource theory of asymmetry, little is known about asymmetry measures in the asymptotic regime. In particular, it is known that the asymptotic rate of metric adjusted skew information is \textit{not} a valid asymptotic asymmetry measure due to a discontinuity in the asymptotic regime \cite{gour_measuring_2009,marvian_operational_2022}. Also, the asymptotic rate of the relative entropy of $G$-frameness becomes trivial for any i.i.d. states since the regularized relative entropy of $G$-frameness vanishes in the asymptotic limit \cite{gour_measuring_2009}. These results show that useful asymptotic asymmetry measures cannot be obtained from a one-shot asymmetry measure just by calculating its asymptotic rate.

In this paper, we introduce a family of valid asymptotic asymmetry measures, which we call the smooth metric adjusted skew information rates. For this purpose, we first introduce the $\epsilon$-smooth metric adjusted skew information parameterized by the smoothness parameter $\epsilon\in(0,1]$ with the smoothing technique \cite{renner_smooth_2004,renner_security_2005}. Then, we define the smooth metric adjusted skew information rates as the sup- and inf-rates of the $\epsilon$-smooth metric adjusted skew information in the limit of $\epsilon\to 0$. We show that they are valid asymptotic asymmetry measures. For an i.i.d. sequence of a pure state, we further show that they are equal to the metric adjusted skew information of the state. Combining these results, we relate the smooth metric adjusted skew information rates to the coherence cost and the distillable coherence, which are central operational quantities in the asymptotic conversion theory. Concretely, by using an argument of Lieb-Yngvason's non-equilibrium thermodynamics \cite{lieb_entropy_2013}, we prove that the smooth metric adjusted skew information rates provide a lower bound of the coherence cost and an upper bound of the distillable coherence.

This paper is organized as follows: In Sec.~\ref{sec:rta_review}, we first review the basics of the resource theory of asymmetry. We also briefly summarize the properties of metric adjusted skew information, including its asymptotic discontinuity. In Sec.~\ref{sec:s_masi}, we introduce a family of the $\epsilon$-smooth metric adjusted skew information and its asymptotic rates. We state our main theorems on the properties of the smooth metric adjusted skew information rates. Theorem~\ref{thm:smooth_masi_rates_valid_measures} shows that the smooth metric adjusted skew information rates are valid asymptotic asymmetry measures. In Theorem~\ref{thm:smooth_masi_rates_iid}, the smooth metric adjusted skew information rates for i.i.d. states are explicitly calculated. We also present a general asymptotic behavior of metric adjusted skew information for states near i.i.d. pure states, which are used to prove Theorem~\ref{thm:smooth_masi_rates_iid}. The details of the proofs are postponed to Appendix. In Sec.~\ref{sec:cost_dist}, we prove inequalities that relate the smooth metric adjusted skew information rates to the coherence cost and the distillable coherence by using Theorems~\ref{thm:smooth_masi_rates_valid_measures} and \ref{thm:smooth_masi_rates_iid}, based on Lieb-Yngvason's non-equilibrium thermodynamics.

\section{Resource theory of asymmetry}\label{sec:rta_review}

\subsection{Definition}
In this subsection, we review the definition of the resource theory of asymmetry. Fundamental elements of any resource theory are free states and free operations that are prepared and performed freely. In the resource theory of asymmetry, they are defined with respect to a symmetry group $G$. We say that a state $\rho$ is symmetric if and only if it is invariant under a group action, i.e., $\rho=U_g\rho U_g^\dag$ for any $g\in G$, where $U_g$ denotes the unitary representation of $g$. A symmetric state can be prepared without access to the reference frame related to the group $G$ \cite{gour_resource_2008} and is considered a free state in the resource theory of asymmetry. Any state that breaks the symmetry is called asymmetric and is regarded as a resource. 

A free operation in the resource theory of asymmetry is the so-called covariant operation $\mathcal{E}$ satisfying $\mathcal{E}(U_g^{\mathrm{in}}(\rho)U_g^{\mathrm{in}\dag})=U_g^{\mathrm{out}}\mathcal{E}(\rho)U_g^{\mathrm{out}\dag}$ for all states $\rho$ and $g\in G$, where $U_g^{\mathrm{in}}$ and $U_g^{\mathrm{out}}$ are unitary representations of $g$ in the input and output systems, respectively. A covariant operation describes a process that can be implemented without access to the reference frame related to the group $G$. By using the covariant Stinespring dilation theorem \cite{keyl_optimal_1999,marvian_mashhad_symmetry_2012}, it can also be interpreted as the dynamics of a system that couples to an ancillary system initially in a symmetric state under a unitary evolution with a conservation law of additive observables associated with $G$.

Although the resource theory of asymmetry is a general framework for investigating the consequences of symmetries described by any group $G$, we here focus on $U(1)$ and $(\mathbb{R},+)$ groups where a unitary representation is given by $U_t=e^{-\ii H t}$ with an observable $H$ on the system. This corresponds to the case where a single additive quantity is conserved. Such an additive quantity can vary depending on the physical situation we are interested in. To simplify the terminologies, we always call it the Hamiltonian of a system in this paper. In this case, a quantum channel $\mathcal{E}$ is covariant if and only if it commutes with the time translation, that is,
\begin{align}
    \mathcal{E}(e^{-\ii H t}\rho e^{\ii Ht})=e^{-\ii H't}\mathcal{E}(\rho)e^{\ii H't},\quad \forall \rho,\,\forall t\in\mathbb{R},
\end{align}
where $H$ and $H'$ are the Hamiltonians of the input and output systems of $\mathcal{E}$. Furthermore, a symmetric state can be understood as a state diagonalized by the energy eigenbasis since $\rho=e^{-\ii Ht}\rho e^{\ii Ht}$ implies $[\rho,H]=0$. Since an asymmetric state $\rho$ satisfies $[\rho,H]\neq 0$, the resource theory of asymmetry with time-translation symmetry is a branch of resource theories investigating the properties of energetic coherence.

For later convenience, we here introduce notations for convertibility in the resource theory of asymmetry. 
The most basic setup in conversion theory is one-shot conversion without error, which has been analyzed, e.g., in \cite{gour_resource_2008,marvian_mashhad_symmetry_2012,yamaguchi_beyond_2022,gour_role_2022}. 
Suppose that for a given state $\rho$ of a system with Hamiltonian $H$ and a given state $\sigma$ of a system with Hamiltonian $H'$, there exists a covariant operation $\mathcal{E}$ such that $\sigma=\mathcal{E}(\rho)$. In this case, we say that $\rho$ is convertible to $\sigma$ and denote
\begin{align}
    (\rho,H)\cov (\sigma,H').
\end{align}
It should be noted that convertibility depends not only on the states $\rho$ and $\sigma$, but also on the Hamiltonians $H$ and $H'$ since the covariance of channels is defined with respect to the time translations generated by them. The binary relation $\cov$ is a preorder since the identity operation is covariant and any product of covariant operations is covariant. 

Since the exact conversion is too restrictive from a practical point of view, conversion with an error is often analyzed. In this paper, we focus on the setup in which the error vanishes in the asymptotic limit. Consider sequences of states $\widehat{\rho}=\{\rho_m\}_m$ and $\widehat{\sigma}=\{\sigma_m\}_m$ of systems with Hamiltonians $\widehat{H}=\{H_m\}_m$ and $\widehat{H'}=\{H_m'\}_m$, respectively. We say that $\widehat{\rho}$ is asymptotically convertible to $\widehat{\sigma}$ by covariant operations if and only if for any $\epsilon\in(0,1]$, there exists a sequence of covariant operations $\{\mathcal{E}_m\}_m$ such that $\limsup_{m\to\infty}D(\mathcal{E}_m(\rho_m),\sigma_m)\leq \epsilon$, where $D$ denotes the trace distance defined by $D(\rho,\sigma)\coloneqq \frac{1}{2}\|\rho-\sigma\|_1$. Here, the covariance of $\mathcal{E}_m$ is defined with respect to $H_m$ and $H_m'$ for each $m$. In this case, we denote 
\begin{align}
    (\widehat{\rho},\widehat{H})\cova (\widehat{\sigma},\widehat{H'}).
\end{align} 
The subscript in $\cova$ indicates that the binary relation is defined for the asymptotic conversion. The asymptotic conversion theory has been analyzed in \cite{gour_resource_2008,marvian_coherence_2020,marvian_operational_2022,yamaguchi_beyond_2022}. We will briefly review their results in Sec.~\ref{sec:cost_dist}, where we relate the smooth metric adjusted skew information rates to the coherence cost and the distillable coherence. Similarly to the one-shot case, the binary relation $\cova$ is a preorder.

\subsection{Metric adjusted skew informations as asymmetry measures}
Resource theories have the advantage of providing a concrete way to quantify a resource. In general, a resource measure is defined with a conversion relation. 
We say that $R(\rho,H)$ is an asymmetry measure (in the one-shot setting) if and only if the following two conditions are satisfied:
\begin{enumerate}[1)]
    \item Monotonicity: If $(\rho,H)\cov (\sigma,H')$, then $R(\rho,H)\geq R(\sigma,H')$.
    \item If $\rho$ is symmetric, then $R(\rho,H)=0$. 
\end{enumerate}
If a measure $R$ further satisfies the converse of the second condition, i.e., if $R(\rho,H)=0$ implies that $\rho$ is symmetric, we say $R$ is faithful. We do not impose faithfulness as one of the minimal requirements for an asymmetry measure, as well as entanglement measures in entanglement theory.

In the one-shot regime, there are several known asymmetry measures \cite{vaccaro_tradeoff_2008,gour_measuring_2009,gour_resource_2008, zhang_detecting_2017,takagi_skew_2019}. Our main focus in this paper is the family of metric adjusted skew information \cite{hansen_metric_2008}.
They are first introduced as quantifiers of non-commutativity of a state and an observable by extending the Wigner-Yanase skew information \cite{wigner_information_1963}, which are later shown to be asymmetry measures \cite{zhang_detecting_2017,takagi_skew_2019}.

The metric adjusted skew information is closely related to the Riemannian geometry of the state space. In classical information geometry, there is a unique Riemannian metric that monotonically contracts under information processing, which is called the Fisher information metric. The study of monotone metrics in quantum information theory is initiated by Morozova and Chentsov in \cite{morozova_markov_1989}. Its classification is completed by Petz in \cite{petz_monotone_1996}, showing a one-to-one correspondence between a monotone metric and an operator monotone function. 
We say that a function $f$ is a standard monotone function if and only if it satisfies the following three conditions:
\begin{enumerate}[i)]
    \item $0\leq A\leq B$ $\implies $ $f(A)\leq f(B)$.
    \item $f(x)=xf(1/x)$
    \item $f(1)=1$
\end{enumerate}
It is known that 
\begin{align}
    f_{\mathrm{RLD}}(x)\leq f(x)\leq f_{\mathrm{SLD}}(x),\quad x>0
\end{align}
for any standard monotone functions $f$ \cite{kubo_means_1980,petz_quantum_2008,gibilisco_refinement_2011}, where $f_{\mathrm{RLD}}\coloneqq 2x/(x+1)$ and $f_{\mathrm{SLD}}\coloneqq (x+1)/2$. Here, RLD and SLD are abbreviations for right logarithmic derivative and symmetric logarithmic derivative. 
For a family of states $\rho_t$ parameterized by a single real number $t\in\mathbb{R}$, the family of quantum Fisher information is defined by the norm of $\partial_t\rho_t$ with respect to the monotone metric \cite{petz_monotone_1996,petz_introduction_2011}, given by
\begin{align}
    J^f(\rho_t)\coloneqq \mathrm{Tr}\left(\frac{\partial\rho_t}{\partial t}c_f\left(L_{\rho_t},R_{\rho_t}\right)\left(\frac{\partial\rho_t}{\partial t}\right)\right),\label{eq:fisher_info_definition}
\end{align}
where $L_\rho$ and $R_\rho$ denote the left and right multiplication operators by $\rho$, i.e., $L_\rho(O)=\rho O$ and $R_\rho(O)=O\rho$ for an operator $O$. The function $c_f(x,y)$ is called the Morozova–Chentsov function associated with a standard monotone function $f$, which is defined by
\begin{align}
    c_f(x,y)=\frac{1}{yf(x/y)} ,\quad x,y>0.\label{eq:mc_function_definition}
\end{align}

A standard monotone function satisfying $f(0)\neq 0$ is called regular. Hansen \cite{hansen_metric_2008} defined the metric adjusted skew information for a regular standard monotone $f$ as
\begin{align}
    I^f(\rho,H)\coloneqq \frac{f(0)}{2}J^f(\rho_t)\biggl|_{t=0}\label{eq:def_skew_info}
\end{align}
where $\rho_t$ denotes a unitary model defined by $\rho_t\coloneqq\{e^{-\ii Ht}\rho e^{\ii Ht}\}_{t\in\mathbb{R}}$.
The prefactor $f(0)/2$ is chosen so that 
\begin{align}
    I^f(\psi,H)=\mathrm{Var}(\psi,H)
\end{align}
holds for any pure state $\psi$, where $\mathrm{Var}(\rho,H)\coloneqq \mathrm{Tr}(\rho H^2)-(\mathrm{Tr}(\rho H))^2$ denotes the variance. For a generic state $\rho$, it holds $I^f(\rho,X)\leq \mathrm{Var}(\rho,X)$. 
By using the eigenvalue decomposition $\rho=\sum_{i}\lambda_i\ket{i}\bra{i}$, the metric adjusted skew informations are written as
\begin{align}
    I^f(\rho,H)=\frac{f(0)}{2}\sum_{i,j}\frac{(\lambda_i-\lambda_j)^2}{\lambda_jf(\lambda_i/\lambda_j)}|\braket{i|H|j}|^2.\label{eq:skew_info_def}
\end{align}

An example of the metric adjusted skew information is the Wigner-Yanase-Dyson skew information, which is known as a one-parameter extension of the Wigner-Yanase skew information. The corresponding operator monotone function is given by $f_{\mathrm{WYD},p}(x)=p(1-p)(x-1)^2/((x^p-1)(x^{1-p}-1))$ \cite{hansen_metric_2008}. The Wigner-Yanase skew information is obtained as a special case for $p=1/2$. Another important example is skew information for $f_{\mathrm{SLD}}$, which corresponds to the SLD Fisher information. It is known that the SLD skew information is equal to the convex roof of variance \cite{toth_extremal_2013,yu_quantum_2013}: $I^{f_{\mathrm{SLD}}}(\rho,H)=\min_{\{p_i,\phi_i\}}\sum_ip_i\mathrm{Var}(\phi_i,H)$, where $\{p_i,\phi_i\}$ runs over the set of all probability distributions $\{p_i\}_i$ and pure states $\{\phi_i\}_i$ satisfying $\rho=\sum_ip_i \phi_i$. The SLD skew information is the largest among the family of metric adjusted skew information.  Precisely, it holds 
\begin{align}
    \frac{1}{2f(0)}I^f(\rho,H)\geq I^{f_{\mathrm{SLD}}}(\rho,H)\geq I^f(\rho,H)\label{eq:inequalities_skew_infos},
\end{align}
where $f$ is an arbitrary regular standard monotone function \cite{gibilisco_inequalities_2009}. 

From the viewpoint of the resource theory of asymmetry, critical properties of metric adjusted skew informations are the following two:
\begin{enumerate}[1)]
    \item Monotonicty \cite{zhang_detecting_2017}: If $(\rho,H)\cov(\sigma,H')$, then $I^f(\rho,H)\geq I^f(\sigma,H')$.
    \item If $\rho$ is symmetric, then $I^f(\rho,H)=0$.
\end{enumerate}
In other words, metric adjusted skew informations are valid asymmetry measures. Furthermore, since $J^f(\rho_t)=0$ only if $\partial_t \rho_t=0$, metric adjusted skew informations are faithful asymmetry measures.

Another property which we shall use later is the convexity of the metric adjusted skew informations: For any set of density operator $\{\rho_k\}$ and probability distribution $\{p_k\}_k$ such that $p_k\geq 0$ and $\sum_k p_k=1$, it holds
\begin{align}
    \sum_k p_k I^f(\rho_k,H)\geq I^f\left(\sum_k p_k \rho_k,H\right).
\end{align}
Such a convexity is one of the requirements that Wigner and Yanase \cite{wigner_information_1963} imposed on quantifiers of information content under conservation laws. The convexity of the Wigner-Yanase skew information is proved in \cite{wigner_information_1963}. The convexity of the Wigner-Yanase-Dyson skew information was called the Wigner-Yanase-Dyson conjecture, which was later proved by Lieb \cite{lieb_convex_1973}. The convexity for all metric adjusted skew informations is proven by Hansen \cite{hansen_metric_2008}.

\subsection{Asymptotic discontinuity of metric adjusted skew informations}
The metric adjusted skew informations are valid asymmetry measures that monotonically decrease under covariant operations. However, their asymptotic rates
\begin{align}
    \limsup_{m\to\infty}\frac{1}{m}I^f(\rho_m,H_m),\quad \liminf_{m\to\infty}\frac{1}{m}I^f(\rho_m,H_m)\label{eq:masi_rates}
\end{align}
are \textit{not} asymptotic asymmetry measures. Here, we say that a quantity $R(\widehat{\rho},\widehat{H})$ defined for sequences of states $\widehat{\rho}$ and Hamiltonians $\widehat{H}$ is an asymptotic asymmetry measure if and only if it satisfies the following two conditions:
\begin{enumerate}[1)]
    \item Monotonicity: If $(\widehat{\rho},\widehat{H})\cova (\widehat{\sigma},\widehat{H'})$, then $R(\widehat{\rho},\widehat{H})\geq R(\widehat{\sigma},\widehat{H'})$.
    \item If $\widehat{\rho}$ is a sequence of symmetric states, then $R(\widehat{\rho},\widehat{H})=0$.
\end{enumerate}
We usually do not impose faithfulness for asymptotic resource measures as a minimal requirement. This is because an operationally important asymmetry measure called the distillable coherence is not faithful. A similar argument can also be found in entanglement theory \cite{horodecki_limits_2000}.

Intuitively, we can understand the reason why the metric adjusted skew information rates in Eq.~\eqref{eq:masi_rates} are not asymptotic asymmetry measures as follows: From Eq.~\eqref{eq:skew_info_def}, we can estimate the maximal change in metric adjusted skew informations as $\sim \|H_m\|^2\times \epsilon$ when the state changes on the order of $\epsilon$ in the trace distance. This implies that even when $\widehat{\rho}$ and $\widehat{\sigma}$ are interconvertible, their metric adjusted skew information rates can take different values, violating condition 1) above. We remark that this was one of the non-trivial issues in establishing the conversion theory in the i.i.d. regime \cite{marvian_operational_2022}. 

Indeed, such a discontinuity can be seen by examining the variance. A quantity $A$ is called asymptotically continuous \cite{donald_uniqueness_2002,plenio_introduction_2007,gour_measuring_2009} if
\begin{align}
    \lim_{m\to\infty}\frac{A(\rho_m)-A(\sigma_m)}{1+\log(\mathrm{dim}(\mathcal{H}_m))}=0\label{eq:asy_continuous}
\end{align}
holds for any states $\rho_m$ and $\sigma_m$ of a Hilbert space $\mathcal{H}_m$ such that $\lim_{m\to\infty}D(\rho_m,\sigma_m)=0$. In a typical setup, including i.i.d. setup, $1+\log(\mathrm{dim}(\mathcal{H}_m))=O(m)$ for $m\to\infty$. Therefore, Eq.~\eqref{eq:asy_continuous} implies $\lim_{m\to\infty}\frac{1}{m}A(\rho_m)=\lim_{m\to\infty}\frac{1}{m}A(\sigma_m)$ if $\lim_{m\to\infty}D(\rho_m,\sigma_m)=0$. 
In \cite{gour_measuring_2009}, it is pointed out that the variance is not asymptotically continuous. Since metric adjusted skew informations are equal to the variance for pure states, this result proves that the metric adjusted skew information rates cannot be asymptotic asymmetry measures.

\section{Smooth metric adjusted skew information rates}\label{sec:s_masi}
So far, we have seen that the asymptotic rates of the metric adjusted skew informations are not asymptotic asymmetry measures. Intuitively, this is because metric adjusted skew informations change too drastically with a small perturbation in the state space and hence they are not good quantifiers in the asymptotic conversion theory with an error. 

To find an asymptotic asymmetry measure related to the metric adjusted skew informations, we here apply the smoothing technique to metric adjusted skew informations. The smoothing technique \cite{renner_smooth_2004,renner_security_2005} is commonly used to investigate information-theoretic tasks with an error. It is closely related to the information-spectrum method \cite{han_information-spectrum_2003}, a powerful and universal tool to analyze asymptotic problems in information theory. The smoothing technique is used in a recent study \cite{yamaguchi_beyond_2022} on the asymptotic conversion theory in the resource theory of asymmetry in the non-i.i.d. regime. In \cite{yamaguchi_beyond_2022}, the smoothing technique is applied for the max- and min-quantum Fisher information to construct an information-spectrum approach for quantum Fisher information. However, it has not been applied to metric adjusted skew informations. 

To begin with, let us introduce a family of the $\epsilon$-smooth metric adjusted skew informations.
For a smoothing parameter $\epsilon\in(0,1]$, we define
\begin{align}
    I_\epsilon^{f}(\rho,H)\coloneqq \inf_{\sigma\in B^\epsilon(\rho)}I^f(\sigma,H),
\end{align}
where $B^{\epsilon}(\rho)$ is the $\epsilon$-ball in the state space defined by $B^{\epsilon}(\rho)\coloneqq \{\sigma: \text{ states}\mid D(\rho,\sigma)\leq \epsilon\}$.
In Appendix~\ref{app:properties_smooth_masi}, we show that $I_\epsilon^{f}(\rho,H)$ monotonically decreases through a covariant channel.
In addition, it trivially vanishes for any symmetric state. 
Therefore, $I_\epsilon^{f}(\rho,H)$ is an asymmetry measure. Note that the $\epsilon$-smooth metric adjusted skew informations $I^f_\epsilon$ is not a faithful measure as it vanishes for a state $\epsilon$-close to a symmetric state.
In Appendix~\ref{app:properties_smooth_masi}, we prove that $I_\epsilon^{f}(\rho,H)$ inherits the convexity of metric adjusted skew informations.

By using the smooth metric adjusted skew informations, we now define the sup- and inf-smooth metric adjusted skew information rates as follows:
\begin{align}
    I_+^{f}(\widehat{\rho},\widehat{H})&\coloneqq \lim_{\epsilon\to 0^+}\limsup_{m\to\infty}\frac{1}{m}I_{\epsilon}^{f}(\rho_m,H_m),\\
    I_-^{f}(\widehat{\rho},\widehat{H})&\coloneqq \lim_{\epsilon\to 0^+}\liminf_{m\to\infty}\frac{1}{m}I_\epsilon^{f}(\rho_m,H_m).
\end{align}
We also call them \textit{the smooth metric adjusted skew information rates} for short. 
We remark that since $\limsup_{m\to\infty}\frac{1}{m}I_{\epsilon}^{f}(\rho_m,H_m)$ and $\liminf_{m\to\infty}\frac{1}{m}I_\epsilon^{f}(\rho_m,H_m)$ are monotonic functions of $\epsilon$, their right-hand side limits $I_\pm^{f}(\widehat{\rho},\widehat{H})$ always exist. 

A key property of the smooth metric adjusted skew information rates is the following:
\begin{theorem}\label{thm:smooth_masi_rates_valid_measures}
For any regular operator monotone function $f$, the smooth metric adjusted skew information rates are asymptotic measures in the resource theory of asymmetry. That is, the following two conditions are satisfied:
\begin{enumerate}[1)]
    \item Monotonicity: If $(\widehat{\rho},\widehat{H})\cova (\widehat{\sigma},\widehat{H'})$, then $I_+^{f}(\widehat{\rho},\widehat{H})\geq I_+^{f}(\widehat{\sigma},\widehat{H'})$ and $I_-^{f}(\widehat{\rho},\widehat{H})\geq I_-^{f}(\widehat{\sigma},\widehat{H'})$. 
    \item If $\widehat{\rho}$ is a sequence of symmetric states, then $I_+^{f}(\widehat{\rho},\widehat{H})=I_-^{f}(\widehat{\rho},\widehat{H})=0$. 
\end{enumerate}
\end{theorem}
This is one of the main results of this paper. The proofs of Theorem~\ref{thm:smooth_masi_rates_valid_measures} and the convexity of smooth metric adjusted skew information rates are provided in Appendix~\ref{app:masi_rate}. 

Note that the monotonicity implies that for any sequence of Hamiltonians $\widehat{H}$, 
\begin{align}
    I^{f}_+(\widehat{\rho},\widehat{H})&=I^{f}_+(\widehat{\sigma},\widehat{H}),\\ I^{f}_-(\widehat{\rho},\widehat{H})&=I^{f}_-(\widehat{\sigma},\widehat{H})
\end{align}
hold if $\lim_{m\to\infty}D(\rho_m,\sigma_m)=0$. 

We remark that from Eq.~\eqref{eq:inequalities_skew_infos}, it holds
\begin{align}
    &\frac{1}{2f(0)}I_+^f(\rho,H)\geq I_+^{f_{\mathrm{SLD}}}(\rho,H)\geq I_+^f(\rho,H),\label{eq:inequalities_skew_info_rates}\\
    &\frac{1}{2f(0)}I_-^f(\rho,H)\geq I_-^{f_{\mathrm{SLD}}}(\rho,H)\geq I_-^f(\rho,H)
\end{align}
for any regular standard monotone function $f$. Therefore, smooth skew information rates for $f_{\mathrm{SLD}}$ are the largest among the family of smooth metric adjusted skew information rates. 

Another main result of this paper is the following:
\begin{theorem}\label{thm:smooth_masi_rates_iid}
Let $\psi$ be a pure state having period $2\pi$ for a Hamiltonian $H$. Assume that the third absolute moment of the Hamiltonian is finite, i.e., $\braket{\psi||H|^3|\psi}<\infty$. For a positive parameter $R>0$, define $\widehat{\psi}_{\mathrm{iid}}(R)\coloneqq \{\psi^{\otimes \ceil{Rm}}\}_m$ and $\widehat{H}_{\mathrm{iid}}(R)\coloneqq \{H_{\mathrm{iid},\ceil{Rm}}\}_m$, where $H_{\mathrm{iid},k}\coloneqq \sum_{i=1}^k \mathbb{I}^{\otimes i-1}\otimes H\otimes\mathbb{I}^{\otimes k-i} $. The smooth metric adjusted skew information rates for this i.i.d. sequence are given by
\begin{align}
    &I_{+}^f(\widehat{\psi}_{\mathrm{iid}}(R),\widehat{H}_{\mathrm{iid}}(R))= I_{-}^f(\widehat{\psi}_{\mathrm{iid}}(R),\widehat{H}_{\mathrm{iid}}(R))\nonumber\\
    &=\lim_{m\to\infty}\frac{1}{m}I^{f}(\psi^{\otimes \ceil{Rm}},H_{\ceil{Rm}})=I^f\left(\psi,H\right)R.
\end{align}
\end{theorem}

This theorem is an immediate corollary of the following lemma, proven in Appendix~\ref{app:smooth_masi_iid}:
\begin{lemma}\label{lem:local_minima_MASI}
Let $\psi$ be a pure state with period $2\pi$ with a Hamiltonian $H$. Assume that the absolute third moment is finite, i.e.,  $\braket{\psi||H|^3|\psi}<\infty$. 
Fix $\epsilon$ to be a sufficiently small real parameter. Let $\widehat{\rho}=\{\rho_m\}_m$ be a sequence of states such that $\rho_m\in B^\epsilon(\psi^{\otimes \ceil{Rm}})$ for all sufficiently large $m$ with a real parameter $R>0$, where $B^{\epsilon}(\rho)$ is the $\epsilon$-ball in the state space defined by $B^{\epsilon}(\rho)\coloneqq \{\sigma: \text{ states}\mid D(\rho,\sigma)\leq \epsilon\}$. Then there exists a real function $\delta^f(\epsilon)$ of $\epsilon$ such that $\lim_{\epsilon\to 0}\delta^f(\epsilon)=0$ and
\begin{align}
    &I^f\left(\rho_m,H_{\mathrm{iid},\ceil{mR}}\right)\nonumber\\
    &\quad \geq I^f\left(\psi^{\otimes \ceil{Rm}},H_{\mathrm{iid},\ceil{mR}}\right)-m \delta^f(\epsilon)+o(m) 
\end{align}
as $m\to\infty$.
\end{lemma}

Intuitively, this lemma shows that $\frac{1}{m}I^f\left(\rho_m,H_{\mathrm{iid},\ceil{mR}}\right)$ approximately takes a local minimum around i.i.d. states  $\psi^{\otimes \ceil{Rm}}$ in the asymptotic regime. As a consequence, its smooth metric adjusted skew information rates become equal to the asymptotic rate of metric adjusted skew informations for i.i.d. states, proving Theorem~\ref{thm:smooth_masi_rates_iid}.

\section{Smooth metric adjusted skew information rates and the asymptotic conversion theory}\label{sec:cost_dist}
In this section, by using Theorems~\ref{thm:smooth_masi_rates_valid_measures} and \ref{thm:smooth_masi_rates_iid}, we relate smooth metric adjusted skew information rates to the coherence cost and the distillable coherence, which are the central quantities in the asymptotic conversion theory in the resource theory of asymmetry.

\subsection{Coherence cost and distillable coherence}
In order to introduce the coherence cost and the distillable coherence, let us briefly review the asymptotic conversion theory in the resource theory of asymmetry. 

Independent and identically distributed (i.i.d.) setup is one of the most important and fundamental regimes in the asymptotic conversion. In this setting, many copies of a state are converted to copies of another state under the assumption that the Hamiltonian is given by a sum of copies of a free Hamiltonian of a subsystem. Precisely, we consider a sequence of states $\widehat{\rho}_{\mathrm{iid}}=\{\rho^{\otimes m}\}_{m}$ with Hamiltonians $\widehat{H}_{\mathrm{iid}}=\{H_m\}_m$, where $H_m=\sum_{i=1}^m  \mathbb{I}^{\otimes (i-1)}\otimes H\otimes \mathbb{I}^{\otimes (m-i)}$.

The convertibility among i.i.d. pure states is studied in earlier research \cite{gour_resource_2008} and completed in \cite{marvian_operational_2022}. Suppose that two pure states $\psi$ and $\phi$ have the same period. For i.i.d. sequences $\widehat{\psi}_{\mathrm{iid}}=\{\psi^{\otimes m}\}$ and $\widehat{\phi}_{\mathrm{iid}}(R)=\{\phi^{\otimes \ceil{Rm}}\}$, it is shown that $\widehat{\psi}_{\mathrm{iid}}$ is convertible to $\widehat{\phi}_{\mathrm{iid}}(R)$ if and only if $R\leq \mathcal{F}(\psi,H)/\mathcal{F}(\phi,H')$. Here, $\mathcal{F}(\rho,H)$ denotes the symmetric logarithmic derivative (SLD) quantum Fisher information with respect to the one-parameter family of states $\rho_t=\{e^{-\ii Ht}\rho e^{\ii Ht}\}_{t\in\mathbb{R}}$, given by
\begin{align}
    \mathcal{F}(\rho,H)\coloneqq 2\sum_{i,j}\frac{(\lambda_i-\lambda_j)^2}{\lambda_i+\lambda_j}|\braket{i|H|j}|^2,
\end{align}
where $\rho=\sum_i\lambda_i\ket{i}\bra{i}$ is the eigenvalue decomposition. Conventionally, this quantity is simply called the quantum Fisher information. Note that the quantum Fisher information is related to the skew information of $f_{\mathrm{SLD}}=(1+x)/2$ as $4I^{f_{\mathrm{SLD}}}(\rho,H)=\mathcal{F}(\rho,H)$.

The above result in the i.i.d. regime shows that pure states with the same periods are equivalent coherence resources in the sense that they are interconvertible with a non-vanishing rate as long as the quantum Fisher information is non-zero. Since the period of a state can be set to $2\pi$ by rescaling the Hamiltonian, we can pick up any pure state with period $2\pi$ as a reference in analyzing convertibility.
Here we adopt a coherence bit
\begin{align}
    \phi_{\mathrm{coh}}\coloneqq \ket{\phi_{\mathrm{coh}}}\bra{\phi_{\mathrm{coh}}},\quad\ket{\phi_{\mathrm{coh}}}=\frac{1}{\sqrt{2}}\left(\ket{0}+\ket{1}\right)\label{eq:cbit_definition}
\end{align}
with Hamiltonian $H_{\mathrm{coh}}=\ket{1}\bra{1}$ as a reference. This state has period $2\pi$ and its quantum Fisher information is normalized to unity: $\mathcal{F}(\phi_{\mathrm{coh}},H_{\mathrm{coh}})=1$. For arbitrary sequences of states $\widehat{\rho}=\{\rho_m\}_{m=1}^\infty$ and Hamiltonians $\widehat{H}=\{H_m\}_{m=1}^\infty$, the coherence cost and the distillable coherence are defined by
\begin{align}
    &C_{\mathrm{cost}}\left(\widehat{\rho},\widehat{H}\right)\nonumber\\
    &\coloneqq \inf\left\{R\mid \left(\widehat{\phi_{\mathrm{coh}}}(R),\widehat{H_{\mathrm{coh}}}(R)\right)\cova \left(\widehat{\rho},\widehat{H}\right)\right\},\\
    &C_{\mathrm{dist}}\left(\widehat{\rho},\widehat{H}\right)\nonumber\\
    &\coloneqq \sup\left\{R\mid \left(\widehat{\rho},\widehat{H}\right)\cova \left(\widehat{\phi_{\mathrm{coh}}}(R),\widehat{H_{\mathrm{coh}}}(R)\right)
    \right\},
\end{align}
where $\widehat{H_{\mathrm{coh}}}(R)\coloneqq \{H_{\mathrm{coh},\ceil{Rm}}\}_m$ with $H_{\mathrm{coh},k}\coloneqq \sum_{i=1}^k\mathbb{I}^{\otimes i-1}\otimes H_{\mathrm{coh}}\otimes \mathbb{I}^{\otimes k-i}$  and $\widehat{\phi_{\mathrm{coh}}}(R)\coloneqq\{ \phi_{\mathrm{coh}}^{\otimes \ceil{Rm}}\}_m$. 

The above result \cite{marvian_operational_2022} for i.i.d. pure states can be restated as
\begin{align}
    C_{\mathrm{cost}}\left(\widehat{\psi}_{\mathrm{iid}},\widehat{H}_{\mathrm{iid}}\right)=C_{\mathrm{dist}}\left(\widehat{\psi}_{\mathrm{iid}},\widehat{H}_{\mathrm{iid}}\right)=\mathcal{F}\left(\psi,H\right),\label{eq:iid_cost_dist}
\end{align}
where the period of $\psi$ is set to be $2\pi$ and $\widehat{\psi}_{\mathrm{iid}}=\{\psi^{\otimes m}\}_m$. This result for the coherence cost is directly extended to a general i.i.d. states $\widehat{\rho}_{\mathrm{iid}}=\{\rho^{\otimes m}\}$ as $ C_{\mathrm{cost}}(\widehat{\rho}_{\mathrm{iid}},\widehat{H}_{\mathrm{iid}})=\mathcal{F}\left(\rho,H\right)$ \cite{marvian_operational_2022}. Although no general formula on the distillable coherence is known for an i.i.d. mixed states, it is known that i.i.d. mixed states typically have vanishing distillable coherence \cite{marvian_coherence_2020}. 

\subsection{Smooth metric adjusted skew information rates, the coherence cost and the distillable coherence}
We are now ready to relate the smooth metric adjusted skew information rates to the coherence cost and the distillable coherence by using Lieb-Yngvason's non-equilibrium thermodynamics \cite{lieb_entropy_2013}. Note that although the original aim of Lieb-Yngvason's non-equilibrium thermodynamics was to extend the notion of thermodynamic entropy to non-equilibrium states, some of their arguments are applicable to a generic preorder relation. 

Concretely, we here prove the following inequalities, proving that the smooth metric adjusted skew information rates yield a lower bound of the coherence cost and an upper bound of the distillable coherence:
\begin{theorem}\label{thm:inequality}
For \textit{any} sequences $\widehat{\rho}$ and $\widehat{H}$, it holds
\begin{align}
    C_{\mathrm{cost}}(\widehat{\rho},\widehat{H})\geq 4I_{+}^f(\widehat{\rho},\widehat{H})\geq 4I_{-}^f(\widehat{\rho},\widehat{H})\geq C_{\mathrm{dist}}(\widehat{\rho},\widehat{H})\label{eq:inequalities_cost_rates_dist}.
\end{align}
\end{theorem}

\begin{proof}

Suppose that $C_{\mathrm{cost}}(\widehat{\rho},\widehat{H})<\infty$. For any $\delta>0$, define $R_\delta\coloneqq C_{\mathrm{cost}}(\widehat{\rho},\widehat{H})+\delta$. Then it holds
\begin{align}
   \left(\widehat{\phi_{\mathrm{coh}}}\left(R_\delta\right),\widehat{H}_{\mathrm{coh}}\left(R_\delta\right)\right)\cova \left(\widehat{\rho},\widehat{H}\right).
\end{align}
From the monotonicity of smooth metric adjusted skew information rates, we get
\begin{align}
    I_{+}^f\left(\widehat{\phi_{\mathrm{coh}}}\left(R_\delta\right),\widehat{H}_{\mathrm{coh}}\left(R_\delta\right)\right)\geq I_{+}^f\left(\widehat{\rho},\widehat{H}\right).
\end{align}
Since
\begin{align}
    I_{+}^f\left(\widehat{\phi_{\mathrm{coh}}}\left(R_\delta\right),\widehat{H}_{\mathrm{coh}}\left(R_\delta\right)\right)=\frac{R_\delta}{4}
\end{align}
holds from Theorem~\ref{thm:smooth_masi_rates_iid}, we get
\begin{align}
    C_{\mathrm{cost}}(\widehat{\rho},\widehat{H})+\delta\geq 4I_{+}^f(\widehat{\rho},\widehat{H})
\end{align}
for any $\delta>0$ and hence
\begin{align}
    C_{\mathrm{cost}}(\widehat{\rho},\widehat{H})\geq 4I_{+}^f(\widehat{\rho},\widehat{H}).
\end{align}
Note that when $C_{\mathrm{cost}}(\widehat{\rho},\widehat{H})>+\infty$, $C_{\mathrm{cost}}(\widehat{\rho},\widehat{H})\geq I_{+}^f(\widehat{\rho},\widehat{H})$ is trivially true. By a similar argument, it is shown that 
\begin{align}
    4 I_{-}^f(\widehat{\rho},\widehat{H})\geq C_{\mathrm{dist}}(\widehat{\rho},\widehat{H}).
\end{align}
\end{proof}

Among these inequalities, the tightest lower bound for the coherence cost can be obtained by $I_{+}^{f_{\mathrm{SLD}}}(\widehat{\rho},\widehat{H})$ as can be directly shown from Eq.~\eqref{eq:inequalities_skew_info_rates}. 

Theorem~\ref{thm:inequality} is of fundamental importance since it clarifies a relation among smooth metric adjusted skew information rates and quantities with operational meaning, i.e., the coherence cost and the distillable coherence. Although the proof provided here is straightforward, the result itself is non-trivial since it is derived from Theorems~\ref{thm:smooth_masi_rates_valid_measures} and \ref{thm:smooth_masi_rates_iid}. To understand its importance, we shall apply Theorem~\ref{thm:inequality} to explicit examples in Sec.~\ref{sec:applications}.

\subsection{Relation to the spectral QFI rates and comparison to entropic quantities}

For pure states with finite periods, asymptotic conversion theory has been established beyond the i.i.d. regime \cite{yamaguchi_beyond_2022}. This result shows that the coherence cost and the distillable coherence for any sequence of pure states are given by the spectral sup- and inf-quantum Fisher information rates, respectively. By using Theorem~\ref{thm:inequality}, it is straightforward to show that the smooth metric adjusted skew information rates give lower and upper bounds for the spectral quantum Fisher information rates. A similar relation is known in entropy rates and the spectral entropy rates. This subsection aims to compare these relations and explore their similarities and differences. 

Let us first briefly review the result in \cite{yamaguchi_beyond_2022}. The definitions of quantities reviewed here are somewhat involved and therefore postponed to Appendix~\ref{app:spectral_qfi}. In \cite{yamaguchi_beyond_2022}, two quantities $\mathcal{F}_{\max}(\psi,H)$ and $\mathcal{F}_{\min}(\psi,H)$, called the max- and min-quantum Fisher information, are introduced for pure states with finite periods. It is proven \cite{yamaguchi_beyond_2022} that they give upper and lower bounds of the quantum Fisher information:
\begin{align}
    \mathcal{F}_{\max}\left(\psi,H\right)\geq \mathcal{F}\left(\psi,H\right)\geq \mathcal{F}_{\min}\left(\psi,H\right).\label{eq:max_min}
\end{align}
For a sequence of pure states with a finite period, the spectral sup- and inf-quantum Fisher information rates are defined by
\begin{align}
    \overline{\mathcal{F}}(\widehat{\psi},\widehat{H})&\coloneqq \lim_{\epsilon\to 0}\limsup_{m\to\infty}\frac{1}{m}\mathcal{F}_{\max}^\epsilon\left(\psi_m,H_m\right),\\
    \underline{\mathcal{F}}(\widehat{\psi},\widehat{H})&\coloneqq \lim_{\epsilon\to 0}\liminf_{m\to\infty}\frac{1}{m}\mathcal{F}_{\min}^\epsilon\left(\psi_m,H_m\right),
\end{align}
where $\mathcal{F}_{\max}^\epsilon$ and $\mathcal{F}_{\min}^\epsilon$ are smooth max- and min-quantum Fisher information. See Appendix~\ref{app:spectral_qfi} for their definition. It is proven \cite{yamaguchi_beyond_2022} that the coherence cost and the distillable coherence are given by
\begin{align}
    C_{\mathrm{cost}}(\widehat{\psi},\widehat{H})=\overline{\mathcal{F}}(\widehat{\psi},\widehat{H}),\quad 
    C_{\mathrm{dist}}(\widehat{\psi},\widehat{H})&=\underline{\mathcal{F}}(\widehat{\psi},\widehat{H}).
\end{align}
Therefore, Theorem~\ref{thm:inequality} implies 
\begin{align}
    \overline{\mathcal{F}}(\widehat{\psi},\widehat{H})\geq 4I^f_+(\widehat{\psi},\widehat{H})\geq 4I^f_-(\widehat{\psi},\widehat{H})\geq  \underline{\mathcal{F}}(\widehat{\psi},\widehat{H}).
\end{align}

Let us now consider a special case with $f_{\mathrm{SLD}}$. 
Since $4I^{f_{\mathrm{SLD}}}(\rho,H)=\mathcal{F}(\rho,H)$ holds for any state $\rho$, we get
\begin{widetext}
\begin{align}
    \overline{\mathcal{F}}(\widehat{\psi},\widehat{H})\geq \lim_{\epsilon\to0}\limsup_{m\to\infty}\frac{1}{m}\mathcal{F}^\epsilon(\psi_m,H_m)\geq \lim_{\epsilon\to0}\liminf_{m\to\infty}\frac{1}{m}\mathcal{F}^\epsilon(\psi_m,H_m)\geq \underline{\mathcal{F}}(\widehat{\psi},\widehat{H})\label{eq:pure_inequalities},
\end{align}
\end{widetext}
where we defined
\begin{align}
    \mathcal{F}^{\epsilon}(\psi,H)\coloneqq \inf_{\rho\in B^\epsilon(\psi)}\mathcal{F}(\rho,H)
\end{align}
These inequalities are interpreted as the asymptotic version of inequalities~\eqref{eq:max_min}.

A similar inequalities are known for the spectral sup- and inf-information rates. Let us first introduce the max-entropy $S_{\mathrm{max}}(\rho)\coloneqq \log \left(\mathrm{rank}\left(\rho\right)\right)$ and the min-entropy $S_{\mathrm{min}}(\rho)\coloneqq -\log\left(\|\rho\|_\infty\right)$,
where $\mathrm{rank}\left(\rho\right)$ and $\|\rho\|_\infty$ denote the rank of $\rho$ and the maximum eigenvalue of $\rho$, respectively.
These quantities provide upper and lower bounds for the von Neumann entropy $S(\rho)\coloneqq -\mathrm{Tr}(\rho\log\rho)$ as
\begin{align}
    S_{\max}(\rho)\geq S(\rho)\geq S_{\min}(\rho).\label{eq:smax_smin}
\end{align}
Defining the smooth max- and min-entropies $S_{\mathrm{max}}^\epsilon(\rho)\coloneqq \inf_{\sigma\in B^\epsilon(\rho)}S_{\mathrm{max}}(\sigma)$ and $S_{\mathrm{min}}^\epsilon(\rho)\coloneqq  \sup_{\sigma\in B^\epsilon(\rho)}S_{\mathrm{min}}(\sigma)$, the spectral sup- and inf-entropy rates $\overline{S}(\widehat{\rho})$ and $\underline{S}(\widehat{\rho})$ of a sequence of states $\widehat{\rho}=\{\rho_m\}_m$ are expressed as
\begin{align}
    \overline{S}(\widehat{\rho})&\coloneqq\lim_{\epsilon\to 0} \limsup_{m\to\infty}\frac{1}{m}S_{\mathrm{max}}^\epsilon(\rho_m),\\
    \underline{S}(\widehat{\rho})&\coloneqq\lim_{\epsilon\to 0} \liminf_{m\to\infty}\frac{1}{m}S_{\mathrm{min}}^\epsilon(\rho_m)
\end{align}
by using the smoothing method \cite{datta_smooth_2009,renner_security_2005}. 
It is known that 
\begin{align}
    \overline{S}(\widehat{\rho})\geq \limsup_{m\to\infty}\frac{1}{m}S(\rho_m)\geq \liminf_{m\to\infty}\frac{1}{m}S(\rho_m)\geq \underline{S}(\widehat{\rho}).\label{eq:inequality_entropy}
\end{align}
hold \cite{datta_min-_2009}. These inequalities can be understood as an asymptotic version of inequalities~\eqref{eq:smax_smin}. 

In the resource theory of entanglement, the spectral entropy rates play an essential role in the conversion theory of non-i.i.d. states. Precisely, for a sequence of bipartite pure states $\widehat{\psi}_{AB}=\{\psi_{AB,m}\}_m$, the entanglement cost and the distillable entanglement are given by the spectral sup- and inf-entropy rates for the sequence of reduced states $\widehat{\rho}\coloneqq \{\rho_{m}\}_m$, where $\rho_{m}\coloneqq \mathrm{Tr}_B(\psi_{AB,m})$, respectively \cite{hayashi_general_2003,bowen_asymptotic_2008}. Furthermore, the sup- and inf-rates of the von Neumann entropy for $\widehat{\rho}$ are asymptotic entanglement measures since the entanglement entropy of a pure state $S_{\mathrm{EE}}(\psi_{AB,m})\coloneqq S(\rho_{m})$ is an entanglement measure that is asymptotically continuous \cite{vidal_entanglement_2000,horodecki_limits_2000,horodecki_entanglement_2001, donald_uniqueness_2002,synak-radtke_asymptotic_2006} as shown by the Fannes inequality \cite{fannes_continuity_1973}. 

It is worth emphasizing that there is a crucial difference between Eqs.~\eqref{eq:pure_inequalities} and \eqref{eq:inequality_entropy} for the asymptotic rates of the quantum Fisher information $\mathcal{F}$ and the entanglement entropy $S_{\mathrm{EE}}$. In Eq.~\eqref{eq:pure_inequalities}, smoothing parameter $\epsilon$ must be included in since the quantum Fisher information $\mathcal{F}$ is not asymptotically continuous \cite{gour_measuring_2009}. On the other hand, in Eq.~\eqref{eq:inequality_entropy}, the smoothing technique is not required for the entanglement entropy $S_{\mathrm{EE}}$ since it is asymptotically continuous \cite{vidal_entanglement_2000,horodecki_limits_2000,horodecki_entanglement_2001, donald_uniqueness_2002,synak-radtke_asymptotic_2006}.
The correspondence is summarized in Table~\ref{tab:comparison}.

\begin{widetext}

\setlength{\tabcolsep}{4.5pt}
\renewcommand{\arraystretch}{1.5}
\begin{threeparttable}[htbp]
    \centering
    \caption{Summary of asymptotic resource measures in the resource theory of entanglement and the resource theory of asymmetry, which are based on entropies and the quantum Fisher information (QFI), respectively. }
    \begin{tabular}{|c||c| c|}\hline
    & Entropy&the quantum Fisher information (QFI)\\\hline\hline
    \multirow{4}{*}{Spectral rates} & sup-rate of smooth max-entropy &sup-rate of smooth max-QFI\tnote{$\mathsection$}\\ &$\overline{S}(\widehat{\rho})=\displaystyle\lim_{\epsilon\to 0}\limsup_{m\to\infty}\frac{1}{m}S_{\max}^\epsilon (\rho_{m})$& $\overline{\mathcal{F}}(\widehat{\psi},\widehat{H})=\displaystyle\lim_{\epsilon\to 0}\limsup_{m\to\infty}\frac{1}{m}\mathcal{F}_{\max}^\epsilon(\psi_{m},H_m)$\\\cline{2-3}
    & inf-rate of smooth min-entropy &inf-rate of smooth min-QFI\tnote{$\mathsection$}\\ &$\underline{S}(\widehat{\rho})=\displaystyle\lim_{\epsilon\to 0}\liminf_{m\to\infty}\frac{1}{m}S_{\min}^\epsilon (\rho_{m})$& $\underline{\mathcal{F}}(\widehat{\psi},\widehat{H})=\displaystyle\lim_{\epsilon\to 0}\liminf_{m\to\infty}\frac{1}{m}\mathcal{F}_{\min}^\epsilon(\psi_{m},H_m)$\\\hline
    \multirow{4}{*}{Asymptotic rates}& sup-rate of entanglement entropy\tnote{*} & sup-rate of smooth QFI\\
    &$\displaystyle\limsup_{m\to\infty}\frac{1}{m}S(\rho_m)$&$\displaystyle\lim_{\epsilon\to0}\limsup_{m\to\infty}\frac{1}{m}\mathcal{F}^\epsilon(\psi_m,H_m)$\\\cline{2-3}
    & inf-rate of entanglement entropy\tnote{*} & inf-rate of smooth QFI\\
    &$\displaystyle\liminf_{m\to\infty}\frac{1}{m}S(\rho_m)$&$\displaystyle\lim_{\epsilon\to0}\liminf_{m\to\infty}\frac{1}{m}\mathcal{F}^\epsilon(\psi_m,H_m)$\\\hline
    \end{tabular}
    \begin{tablenotes}
    \item[*]Smoothing technique is \textit{not} required.
    \item[$\mathsection$]Defined for a sequence of pure states with a finite period. 
    \end{tablenotes}
    \label{tab:comparison}
\end{threeparttable}

\end{widetext}

\section{Applications}\label{sec:applications}
To better understand our results established in the present paper, we here apply them to concrete examples. 

\subsection{Upperbound for distillable coherece}
Distillable coherence is particularly relevant for mixed states since it quantifies the number of pure resources that can be extracted from noisy resources. In \cite{marvian_coherence_2020}, the following sufficient condition for the distillable coherence to vanish is proven: 
\begin{align}
    [\Pi_\rho,H]=0\implies C_{\mathrm{dist}}\left(\widehat{\rho}_{\mathrm{iid}},\widehat{H}_{\mathrm{iid}}\right)=0,\label{eq:marvian_dist_iid_mixed}
\end{align}
where 
$\Pi_\rho$ is the projector to the support of $\rho$, $\widehat{\rho}_{\mathrm{iid}}=\{\rho^{\otimes m }\}_m$ and $\widehat{H}_{\mathrm{iid}}=\{H_{\mathrm{iid},m}\}_m$ for $H_{\mathrm{iid},m}\coloneqq \sum_{i=1}^m\mathbb{I}^{\otimes i-1}\otimes H\otimes \mathbb{I}^{\otimes m-i}$. However, despite its importance, no general formula is known for the distillable coherence of mixed states, even in the i.i.d. regime. 

We here derive a simple upper bound for the distillable coherence applicable to general sequence of states. From Theorem~\ref{thm:inequality}, the distillable coherence $C_{\mathrm{dist}}\left(\widehat{\rho},\widehat{H}\right)$ is upper bounded by $4I_-^f\left(\widehat{\rho},\widehat{H}\right)$. From the definition of smooth metric adjusted skew information inf-rate, it holds $\liminf_{m\to\infty}\frac{1}{m}I_-^f(\rho_m,H_m)\geq I^f\left(\widehat{\rho},\widehat{H}\right)$. Therefore, 
\begin{align}
    4\liminf_{m\to\infty}\frac{1}{m}I^f(\rho_m,H_m)\geq C_{\mathrm{dist}}\left(\widehat{\rho},\widehat{H}\right)\label{eq:upper_bound_nonsmooth}
\end{align}
holds for any $\left(\widehat{\rho},\widehat{H}\right)$.

To simplify the argument, let us now consider the i.i.d. case. Since the metric adjusted skew informations are additive for product states \cite{hansen_metric_2008}, we have
\begin{align}
    \liminf_{m\to\infty}\frac{1}{m}I^f(\rho^{\otimes m},H_{\mathrm{iid},m})=I^f(\rho,H),
\end{align}
implying that
\begin{align}
   4 I^f(\rho,H)\geq C_{\mathrm{dist}}\left(\widehat{\rho}_{\mathrm{iid}},\widehat{H}_{\mathrm{iid}}\right).\label{eq:upperbound_dist_iid_case}
\end{align}
Note that for $f_{\mathrm{SLD}}$, this bound is equivalent to a trivial inequality
\begin{align}
    C_{\mathrm{cost}}\left(\widehat{\rho}_{\mathrm{iid}},\widehat{H}_{\mathrm{iid}}\right)\geq C_{\mathrm{dist}}\left(\widehat{\rho}_{\mathrm{iid}},\widehat{H}_{\mathrm{iid}}\right)\label{eq:trivial_bound}
\end{align}
since $C_{\mathrm{cost}}\left(\widehat{\rho}_{\mathrm{iid}},\widehat{H}_{\mathrm{iid}}\right)=4I^{f_{\mathrm{SLD}}}(\rho,H)$ if $\rho$ has a period $2\pi$ \cite{marvian_operational_2022}. Our bound in Eq.~\eqref{eq:upperbound_dist_iid_case} gives a better bound in general since $I^{f_{\mathrm{SLD}}}(\rho,H)$ is the largest element in the family of metric adjusted skew information. 

As a concrete example, let us analyze an i.i.d. copies of qutrit systems with Hamiltonian $H=\sum_{n=0}^2n\ket{n}\bra{n}$. Consider a mixed state
\begin{align}
    \rho\coloneqq (1-q) \psi_1+q\psi_2\label{eq:example_mixed_state}
\end{align}
where $q\in(0,1)$ and $\psi_i\coloneqq \ket{\psi_i}\bra{\psi_i}$ for orthonormal states $\{\ket{\psi_i}\}_{i=1}^2$. From the definition, the metric adjusted skew information is explicitly calculated as
\begin{align}
    &I^f(\rho,H)\nonumber\\
    &=(1-q)\mathrm{Var}(\psi_1,H)+q\mathrm{Var}(\psi_2,H)\nonumber\\
    &\quad -\left(1-\frac{f(0)(1-2q)^2}{qf((1-q)/q)}\right)|\braket{\psi_1|H|\psi_2}|^2
\end{align}
for the mixed state in Eq.~\eqref{eq:example_mixed_state}. 

To compare the bounds in Eqs.~\eqref{eq:upperbound_dist_iid_case} and \eqref{eq:trivial_bound}, we calculate the metric adjusted skew information for
\begin{align}
    \ket{\psi_1}&\coloneqq \frac{1}{2}\ket{0}+\frac{1}{2}\ket{1}+\frac{1}{\sqrt{2}}\ket{2},\\
    \ket{\psi_2}&\coloneqq \frac{1}{\sqrt{2}}\ket{0}-\frac{1}{\sqrt{2}}\ket{1}
\end{align}
Note that the result in Eq.~\eqref{eq:marvian_dist_iid_mixed} is not applicable since $[\Pi_\rho,H]=[\psi_1+\psi_2,H]\neq 0$. For this example, the metric adjusted skew information for $f_{\mathrm{SLD}}$ and $f_{\mathrm{WYD},p}$ are evaluated as
\begin{align}
    &I^{f_{\mathrm{SLD}}}(\rho,H)=\frac{11-15q+8q^2}{16},\\
    &I^{f_{\mathrm{WYD},p}}(\rho,H)=\frac{11-7q}{16}\nonumber\\
    &-\frac{1}{8}\left(1-q\left(\left(\frac{1-q}{q}\right)^p-1\right)\left(\left(\frac{1-q}{q}\right)^{1-p}-1\right)\right)
\end{align}
for $q\in(0,1)$.
Since $I^{f_{\mathrm{WYD},p}}$ is a monotonically increasing continuous function of $p$ for $p\in(0,1/2)$ and $q\in(0,1)$, Eq.~\eqref{eq:upperbound_dist_iid_case} implies
\begin{align}
    &4\lim_{p\to 0^+}I^{f_{\mathrm{WYD},p}}(\rho,H)\biggl|_{q\neq 0,1}\nonumber\\
    &=\frac{9-7q}{4}\geq C_{\mathrm{dist}}\left(\widehat{\rho}_{\mathrm{iid}},\widehat{H}_{\mathrm{iid}}\right)\biggl|_{q\neq 0,1}\label{eq:limit_bound}
\end{align}
for $q\in(0,1)$.

The metric adjusted skew information for $f_{\mathrm{SLD}}$ and $f_{\mathrm{WYD},p}$, and the bound in Eq.~\eqref{eq:limit_bound} are plotted in Fig.~\ref{fig:comparison_skew_info}. Of particular interest are their behaviors near $q\approx 0$ and $q\approx 1$. When $q=0$, i.e., $\rho=\psi_1$, the distillable coherence is given by
\begin{align}
    C_{\mathrm{dist}}\left(\widehat{\rho}_{\mathrm{iid}},\widehat{H}_{\mathrm{iid}}\right)\biggl|_{q=0}=4I^{f_{\mathrm{SLD}}}(\psi_1,H)=\frac{11}{4},\label{eq:dist_0}
\end{align}
while when $q=1$, it holds
\begin{align}
    C_{\mathrm{dist}}\left(\widehat{\rho}_{\mathrm{iid}},\widehat{H}_{\mathrm{iid}}\right)\biggl|_{q=1}=4I^{f_{\mathrm{SLD}}}(\psi_2,H)=1.\label{eq:dist_1}
\end{align}
From Eqs.~\eqref{eq:limit_bound}, \eqref{eq:dist_0} and \eqref{eq:dist_1}, we find that the distillable coherence $C_{\mathrm{dist}}\left(\widehat{\rho}_{\mathrm{iid}},\widehat{H}_{\mathrm{iid}}\right)$ is discontinuous at $q=0,1$. 

\begin{figure}
    \centering
    \includegraphics[width=8cm]{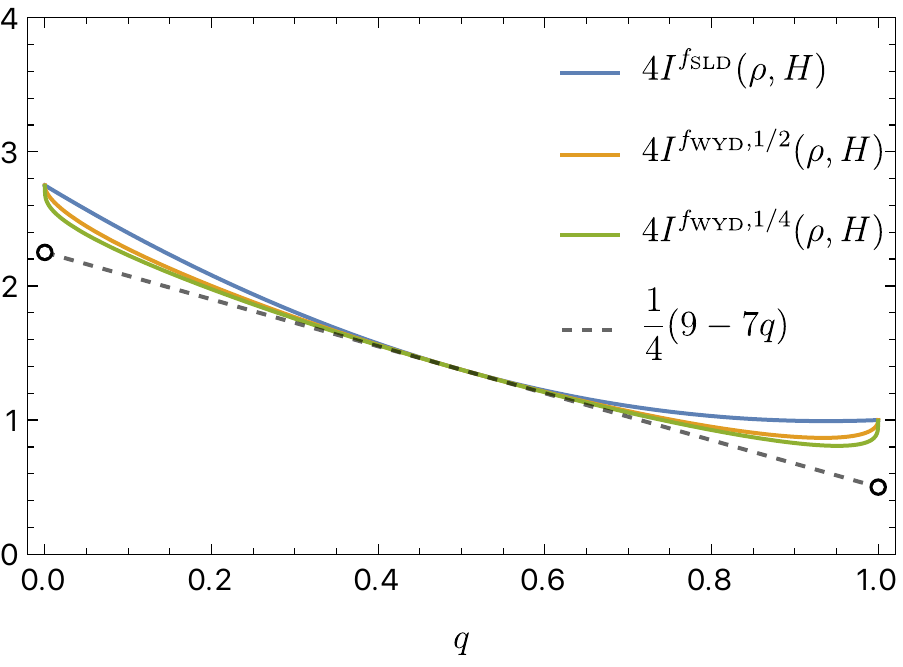}
    \caption{Comparison of upper bounds of the distillable coherence for an i.i.d. mixed state. $4I^{f_{\mathrm{SLD}}}$ corresponds to a trivial bound in Eq.~\eqref{eq:trivial_bound}, while $4I^{f_{\mathrm{WYD},p}}$ is a special case of our new bound in Eq.~\eqref{eq:upperbound_dist_iid_case}. The plot shows that the latter gives a tighter bound. }
    \label{fig:comparison_skew_info}
\end{figure}

We stress that the bounds in this subsection is derived by using Eq.~\eqref{eq:upper_bound_nonsmooth}, which can be calculated without using any smoothing technique. Nevertheless, the smoothing technique is essential to prove Eq.~\eqref{eq:upper_bound_nonsmooth} as we have argued in this paper. In the next subsection, we will revisit this point by analyzing an explicit example. 

\subsection{Necessity of smoothing}
Let us now consider a sequence of pure states $\widehat{\psi}=\{\psi_m\}_m$ for qutrit systems given by
\begin{align}
   \ket{\psi_m}&\coloneqq \sqrt{1-\epsilon_m}\ket{\phi_{\mathrm{coh}}}^{\otimes m}+\sqrt{\epsilon_m}\ket{2}^{\otimes m}
\end{align}
where $\ket{\phi_{\mathrm{coh}}}$ is defined in Eq.~\eqref{eq:cbit_definition} and $\epsilon_m\coloneqq 1/\sqrt{m}$. 
Note that $\widehat{\psi}\cova \widehat{\phi}_{\mathrm{coh}}(1)$ and $\widehat{\phi}_{\mathrm{coh}}(1)\cova \widehat{\psi}$ since $\lim_{m\to\infty}D\left(\psi_m,\phi_{\mathrm{coh}}^{\otimes m}\right)=0$. Therefore, we get
\begin{align}
    C_{\mathrm{cost}}(\widehat{\psi},\widehat{H})=C_{\mathrm{cost}}(\widehat{\psi},\widehat{H})=1.\label{eq:example_cost_dist}
\end{align}

The left-hand side of Eq.~\eqref{eq:upper_bound_nonsmooth} is evaluated as
\begin{align}
    4\lim_{m\to\infty}\frac{1}{m}I^f(\psi_m,H_m)=\infty
\end{align}
since
\begin{align}
    &I^f(\psi_m,H_m)=\mathrm{Var}(\psi_m,H_m)\nonumber\\
    &=\frac{\sqrt{m}}{4}(9m-8\sqrt{m}-1).
\end{align}
Therefore, Eq.~\eqref{eq:upper_bound_nonsmooth} does not provide a meaningful bound for distillable coherence in this example. This can be understood from the fact that the asymptotic rate of metric adjusted skew information is not a good asymptotic asymmetry measure due to the asymptotic discontinuity \cite{gour_measuring_2009,marvian_coherence_2020}.

In the present paper, we showed that the smoothing technique is essential to overcome this issue. From Lemma~\ref{lem:local_minima_MASI}, the smooth metric adjusted skew information rates are explicitly evaluated as
\begin{align}
    4I^f_\pm (\widehat{\psi},\widehat{H})=1.
\end{align}
From Theorem~\ref{thm:inequality}, this quantity gives a lower bound for the coherence cost and an upper bound for the distillable coherence, as is consistent with Eq.~\eqref{eq:example_cost_dist}.

\section{Conclusions}
In this paper, we investigated the properties of metric adjusted skew information particularly from the perspective of the resource theory of asymmetry. A family of metric adjusted skew information is induced from a family of monotone metrics, i.e., quantum Fisher information metrics in information geometry. They are known as asymmetry resource measure in the one-shot regime. However, their asymptotic rates are not valid asymptotic asymmetry monotone due to their asymptotic discontinuity. To find asymptotic asymmetry measures, we first introduced a family of the $\epsilon$-smooth metric adjusted skew information by using the smoothing technique. We then defined the sup- and inf-smooth metric adjusted skew information rates as its asymptotic rates in the limit of $\epsilon\to0$. We proved that these families of smooth metric adjusted skew information rates are valid asymptotic asymmetry measures. We further proved that the smooth metric adjusted skew information rates inherit the convexity of skew informations, which may be viewed as an extension of the Wigner-Yanase-Dyson conjecture \cite{wigner_information_1963,lieb_convex_1973} to the asymptotic regime.

By analyzing a general asymptotic behavior of the metric adjusted skew information for sequences of states that are close to i.i.d. pure states, we explicitly calculated the smooth metric adjusted skew information for an arbitrary sequence of i.i.d pure states. Combining this result with the asymptotic monotonicity of the smooth metric adjusted skew information rates, we related them to the key operational quantities in the asymptotic conversion theory in the resource theory of asymmetry. Concretely, we showed that the smooth metric adjusted skew information rates provide a lower bound of the coherence cost and an upper bound of the distillable coherence. As a corollary, we further proved inequalities relating the asymptotic rates of the quantum Fisher information to the spectral quantum Fisher information rates based on an information-spectrum approach for the quantum Fisher information \cite{yamaguchi_beyond_2022}. They have a structure similar to the inequalities for the von Neumann entropy rates and the spectral entropy rates. However, our analysis shows that smoothing parameters must be included in the rates of the quantum Fisher information since it has an asymptotic discontinuity. 

The main focus of this paper is to study the metric adjusted skew information in the resource theory of asymmetry. However, our results could also be useful in quantum thermodynamics and quantum estimation theory. In the resource theory of athermality, thermal operations are regarded as free. Here, thermal operations are operations that can be implemented by coupling a system to a thermal bath through an energy-conserving unitary evolution \cite{janzing_thermodynamic_2000,brandao_resource_2013,horodecki_fundamental_2013}. Since any thermal operation is covariant with respect to the time translation, Theorem~\ref{thm:smooth_masi_rates_valid_measures} implies that the smooth metric adjusted skew information rates $I^f_{\pm}$ are valid asymptotic resource measures in the resource theory of athermality. Although relevant quantities in quantum thermodynamics are typically introduced from entropies, new insights are obtained with the family of quantum Fisher information \cite{marvian_coherence_2020}. Therefore, it will be quite interesting to explore the role of the smooth metric adjusted skew information rates not only as asymmetry measures but also as athermality measures. On the other hand, the smooth metric adjusted skew information rates can also be helpful in studies on quantum estimation theory in the asymptotic regime. Indeed, the metric adjusted skew information is proportional to quantum Fisher information, a central quantifier in quantum estimation theory, for a unitary model. Further research in these directions is left for future work since it is out of the scope of the present paper.

\begin{acknowledgments}
KY acknowledges support from the JSPS Overseas Research Fellowships. 
HT acknowledges supports from JSPS Grants-in-Aid for Scientific Research (JP19K14610 and JP22H05250), JST PRESTO (JPMJPR2014), and JST MOONSHOT (JPMJMS2061).
\end{acknowledgments}

\bibliographystyle{apsrev4-1}

\bibliography{references}


\onecolumn
\newpage
\appendix
\section{Properties of the $\epsilon$-smooth skew informations}\label{app:properties_smooth_masi}
Let us first prove the monotonicity of the $\epsilon$-smooth skew informations. 
\begin{proposition}[Monotonicity of the $\epsilon$-smooth metric adjusted skew information]
Let $\epsilon$ be a parameter such that $\epsilon\in(0,1]$. Let $(\rho,H)$ and $(\sigma,H')$ denote sets of states and Hamiltonians. If $\rho$ is convertible to $\sigma$ by a covariant operation, i.e., $(\rho,H)\cov (\sigma,H')$, then it holds
\begin{align}
    I_{\epsilon}^f(\rho,H)\geq I_{\epsilon}^f(\sigma,H').
\end{align}
\end{proposition}
\begin{proof}
The proof directly follows from the monotonicity of the metric adjusted skew informaiton and the trace distance. By using a covariant channel $\mathcal{E}$ such that $\mathcal{E}(\rho)=\sigma$, we get
\begin{align}
     I_\epsilon^f(\rho,H)&=\inf_{\xi\in B^\epsilon(\rho)}I_\epsilon^f(\xi,H)\\
    &\geq \inf_{\xi\in B^\epsilon(\rho)}I_\epsilon^f(\mathcal{E}(\xi),H')
\end{align}
from the monotonicity of metric adjusted skew informations.
Since the trace distance is contractive under a quantum channel, $\mathcal{E}(\xi)\in B^{\epsilon}(\mathcal{E}(\rho))$ holds for any $\xi\in B^\epsilon(\rho)$, implying that
\begin{align}
    \inf_{\xi\in B^\epsilon(\rho)}I_\epsilon^f(\mathcal{E}(\xi),H')\geq \inf_{\chi\in B^\epsilon(\sigma)}I_\epsilon^f(\chi,H').
\end{align}
Therefore, we get
\begin{align}
    I_\epsilon^f(\rho,H)\geq I_\epsilon^f(\sigma,H').
\end{align}

\end{proof}

We now prove the convexity of the $\epsilon$-smooth metric adjusted skew informations.
\begin{proposition}[Convexity of the $\epsilon$-smooth metric adjusted skew information]
Let $\epsilon$ be a parameter $\epsilon\in(0,1]$. For any states $\{\rho^{(k)}\}_k$ and a probability distribution $\{p_k\}_k$, it holds
\begin{align}
    \sum_k p_k I_\epsilon^f(\rho_k,H)\geq I_\epsilon^f\left(\sum_k p_k \rho_k,H\right).
\end{align}
\end{proposition}
\begin{proof}
By using the convexity of the metric adjusted skew information
\begin{align}
    \sum_k p_kI^f(\rho_k,H)&\geq I^f\left(\sum_k p_k\rho_k,H\right),
\end{align}
we get
\begin{align}
    \sum_{k}p_k I_{\epsilon}^f(\rho_k,H)&=\sum_{k}p_k\inf_{\sigma_k\in B^\epsilon(\rho_k)}I^f(\sigma_k,H)\\
    &=\inf_{\forall k,\,\sigma_k\in B^\epsilon(\rho_k)}\sum_{k}p_kI^f(\sigma_k,H)\\
    &\geq \inf_{\forall k,\,\sigma_k\in B^\epsilon(\rho_k)}I^f\left(\sum_{k}p_k\sigma_k,H\right)\\
    &\geq \inf_{\sigma\in B^\epsilon(\sum_{k}p_k\rho_k)}I^f(\sigma,H)\\
    &=I_\epsilon^f\left(\sum_k p_k \rho_k,H\right).
\end{align}
In the last inequality, we have used 
\begin{align}
    \forall k,\quad  \sigma_k\in B^\epsilon (\rho_k)\implies \sum_{k}p_k\sigma_k\in B^\epsilon\left( \sum_kp_k\rho_k\right),
\end{align}
which follows from the convexity of the trace distance. See, e.g., \cite{nielsen_quantum_2010} for proof of the convexity of the trace distance.
\end{proof}

\section{Properties of smooth metric adjusted skew information rates}\label{app:masi_rate}

Let us first prove the asymptotic monotonicity of the smooth metric adjusted skew information rates. 
\begin{proposition}[Monotonicity of the smooth metric adjusted skew information rates]
Let $(\widehat{\rho},\widehat{H})$ and $(\widehat{\sigma},\widehat{H}')$ be sequences of states and Hamiltonians. If $(\widehat{\rho},\widehat{H})\cova (\widehat{\sigma},\widehat{H}')$, then it holds
\begin{align}
    I_{\pm}^f(\widehat{\rho},\widehat{H})\geq I_{\pm}^f(\widehat{\sigma},\widehat{H'}).
\end{align}
\end{proposition}
\begin{proof}

Since $(\widehat{\rho},\widehat{H})\cova (\widehat{\sigma},\widehat{H}')$, for any $\epsilon>0$, there exists a sequence of covariant channels $\{\mathcal{E}_m\}_m$ such that 
\begin{align}
    D\left(\mathcal{E}_m(\rho_m),\sigma_m\right)<\epsilon
\end{align}
holds for all sufficiently large $m$. 
By using the triangle inequality, we have
\begin{align}
    \forall \chi_m\in B^\epsilon\left(\mathcal{E}_m(\rho_m)\right),\quad D\left(\chi_m,\sigma_m\right)\leq D\left(\chi_m,\mathcal{E}(\rho_m)\right)+D\left(\mathcal{E}_m(\rho_m),\sigma_m\right)<2\epsilon
\end{align}
for all sufficiently large $m$.
Therefore, we get
\begin{align}
    I_{\epsilon}^f\left(\rho_m,H_m\right)&\geq I_{\epsilon}^f\left(\mathcal{E}_m(\rho_m),H_m'\right)\\
    &\geq I_{2\epsilon}^f\left(\sigma_m,H_m'\right)
\end{align}
and hence
\begin{align}
    \limsup_{m\to\infty}\frac{1}{m}I_{\epsilon}^f\left(\rho_m,H_m\right)&\geq \limsup_{m\to\infty}\frac{1}{m}I_{2\epsilon}^f\left(\sigma_m,H_m'\right),\\
    \liminf_{m\to\infty}\frac{1}{m}I_{\epsilon}^f\left(\rho_m,H_m\right)&\geq \liminf_{m\to\infty}\frac{1}{m}I_{2\epsilon}^f\left(\sigma_m,H_m'\right)
\end{align}
holds for any $\epsilon\in(0,1]$. In the limit of $\epsilon\to 0$, it is proven
\begin{align}
    I_+^f(\widehat{\rho},\widehat{H})&=\lim_{\epsilon\to0^+}\limsup_{m\to\infty}\frac{1}{m}I_{\epsilon}^f\left(\rho_m,H_m\right)\\
    &\geq \lim_{\epsilon\to0^+}\limsup_{m\to\infty}\frac{1}{m}I_{2\epsilon}^f\left(\sigma_m,H_m'\right)\\
    &=I_+^f(\widehat{\sigma},\widehat{H'}),
\end{align}
and similarly
\begin{align}
     I_-^f(\widehat{\rho},\widehat{H})\geq I_-^f(\widehat{\sigma},\widehat{H'}).
\end{align}

\end{proof}

Let us now prove the convexity of the smooth metric adjusted skew information rates. 
\begin{proposition}[Convexity of the smooth metric adjusted skew information rates]
Let $\{p_k\}_k$ be a probability distribution. We denote $\widehat{\rho^{(k)}}$ a sequences of states for each $k$. For a sequence of states $\sum_{k=1}^Np_k \widehat{\rho^{(k)}}$ given by $\{\sum_{k=1}^Np_k \rho^{(k)}_m\}_m$, it holds
\begin{align}
    \sum_{k=1}^Np_kI_{\pm}^f\left(\widehat{\rho}^{(k)},\widehat{H}\right)\geq I_\pm^f\left(\sum_{k=1}^Np_k\widehat{\rho}^{(k)},\widehat{H}\right)
\end{align}
\end{proposition}
\begin{proof}
By using the convexity of smooth metric adjusted skew informations, we immediately get
\begin{align}
    \sum_{k=1}^Np_kI_{+}^f\left(\widehat{\rho}^{(k)},\widehat{H}\right)&=\sum_{k=1}^Np_k\lim_{\epsilon\to 0}\limsup_{m\to\infty}\frac{1}{m}I^f_\epsilon\left(\rho_m^{(k)},H_m\right)\\
    &= \lim_{\epsilon\to 0}\limsup_{m\to\infty}\frac{1}{m}\sum_{k=1}^Np_kI^f_\epsilon\left(\rho_m^{(k)},H_m\right)\\
    &\geq \lim_{\epsilon\to 0}\limsup_{m\to\infty}\frac{1}{m}I^f_\epsilon\left(\sum_{k=1}^Np_k\rho_m^{(k)},H_m\right)\\
    &=I_+^f\left(\sum_{k=1}^Np_k\widehat{\rho}^{(k)},\widehat{H}\right).
\end{align}
It should be noted that this proof is valid only if we can exchange the order of the sum $\sum_{k}$ and the limits $\lim_{\epsilon\to 0}\limsup_{m\to\infty}$. When the summation $\sum_{k}$ is taken over a finite set, this condition is always satisfied. 

In the same way, we also get $\sum_{k=1}^Np_kI_{-}^f\left(\widehat{\rho}^{(k)},\widehat{H}\right)\geq I_-^f\left(\sum_{k=1}^Np_k\widehat{\rho}^{(k)},\widehat{H}\right)$. 
\end{proof}

Combining the convexity of the smooth metric adjusted skew information rates and Theorem~\ref{thm:inequality}, we can derive the following simple upper bound for the distillable coherence:
\begin{align}
    \sum_{k}p_kI_{\pm}^f\left(\widehat{\rho}_k,\widehat{H}\right)\geq I_\pm^f\left(\sum_{k}p_k\widehat{\rho}^{(k)},\widehat{H}\right)\geq  C_{\mathrm{dist}}\left(\sum_{k}p_k\widehat{\rho}^{(k)},\widehat{H}\right).
\end{align}
As a corollary, we also get upper bounds that can be calculated without using the smoothing technique:
\begin{align}
    \liminf_{m\to\infty}\frac{1}{m}\sum_k p_kI^f\left(\rho^{(k)}_m,H_m\right)\geq \liminf_{m\to\infty}\frac{1}{m}I^f\left(\sum_k p_k\rho^{(k)}_m,H_m\right)\geq C_{\mathrm{dist}}\left(\sum_{k}p_k\widehat{\rho}^{(k)},\widehat{H}\right).
\end{align}

\section{Smooth metric adjusted skew information rates in the i.i.d. regime}\label{app:smooth_masi_iid}
We here explicitly calculate the smooth metric adjusted skew information rates in the i.i.d. regime. For this purpose, we here prove
Lemma~\ref{lem:local_minima_MASI}:
\begin{lemma*}[Restatement of Lemma \ref{lem:local_minima_MASI}]
Let $\psi$ be a pure state with period $2\pi$ with a Hamiltonian $H$. Assume that the absolute third moment is finite, i.e.,  $\braket{\psi||H|^3|\psi}<\infty$. 
Fix $\epsilon$ be a sufficiently small real parameter. Let  $\widehat{\rho}=\{\rho_m\}_m$ be a sequence of states such that $\rho_m\in B^\epsilon(\psi^{\otimes \ceil{Rm}})$ for all sufficiently large $m$ with a real parameter $R>0$, where $B^{\epsilon}(\rho)$ is the $\epsilon$-ball in the state space defined by $B^{\epsilon}(\rho)\coloneqq \{\sigma: \text{ states}\mid D(\rho,\sigma)\leq \epsilon\}$. Then there exists a real function $\delta^f(\epsilon)$ of $\epsilon$ such that $\lim_{\epsilon\to 0}\delta^f(\epsilon)=0$ and
\begin{align}
    &I^f\left(\rho_m,H_{\mathrm{iid},\ceil{mR}}\right)\geq I^f\left(\psi^{\otimes \ceil{Rm}},H_{\mathrm{iid},\ceil{mR}}\right)-m \delta^f(\epsilon)+o(m) \quad \quad (m\to\infty),
\end{align}
where $H_{\mathrm{iid},k}\coloneqq \sum_{i=1}^k\mathbb{I}^{\otimes i-1}\otimes H\otimes \mathbb{I}^{\otimes k-i}$.
\end{lemma*}

To prove this lemma, let us first derive a general behavior of metric adjusted skew informations for states close to a pure state. The following lemma is an extension of claims in Corollary~1 in \cite{marvian_coherence_2020}.
\begin{lemma}\label{lem:masi_lowerbound}
Let $\Phi$ denote the eigenstate with the largest eigenvalue of an arbitrary state $\rho$. For a pure state $\Psi$, the infidelity of $\rho$ and $\Psi$ is defined by $\delta\coloneqq 1-  \Braket{\Psi|\rho|\Psi}$. Then it holds
\begin{align}
    |\braket{\Phi|\Psi}|\geq 1-2\delta\label{eq:infidelity_largest_ev}.
\end{align}
Furthermore, when $\delta<1/2$, the metric adjusted skew informations $I^f(\rho,H)$ are lower bounded as 
\begin{align}
    I^f(\rho,H)\geq \frac{f(0)}{f\left(\frac{\delta}{1-\delta}\right)}(1-2\delta)^2\mathrm{Var}(\Phi,H).\label{eq:masi_lowerbound_infidelity}
\end{align}
\end{lemma}

\begin{proof}
Equation~\eqref{eq:infidelity_largest_ev} is proven in Corollary~1 in \cite{marvian_coherence_2020}. We here prove Eq.~\eqref{eq:masi_lowerbound_infidelity}. Let
\begin{align}
    \rho=p\ket{\Phi}\bra{\Phi}+\sum_{i=1}^k\lambda_i\ket{i}\bra{i}\eqqcolon\sum_{i=0}^k\lambda_k\ket{i}\bra{i}
\end{align}
be the eigenvalue decomposition, where we defined $\lambda_0\coloneqq p$ and $\ket{0}\coloneqq\ket{\Phi}$. 
Since 
\begin{align}
    \braket{\Psi|\rho|\Psi}=\sum_{i=0}^k\lambda_k|\braket{\Psi|i}|^2\leq p\sum_{i=0}^k|\braket{\Psi|i}|^2=p, 
\end{align}
we have
\begin{align}
     1-\delta\leq p.
\end{align}
By using $1=p+\sum_{i=1}^k\lambda_i$, we get $\lambda_i\leq 1-p\leq \delta$ for all $i=1,2,\cdots, k$.

Equation~\eqref{eq:masi_lowerbound_infidelity} is derived as
\begin{align}
    I^f(\rho,H)&=\frac{f(0)}{2}\sum_{i,j=0}^k\frac{(\lambda_i-\lambda_j)^2}{\lambda_jf(\lambda_i/\lambda_j)}|\braket{i|H|j}|^2\\
    &=2\times \frac{f(0)}{2}\sum_{i=1}^k\frac{(\lambda_i-p)^2}{pf(\lambda_i/p)}|\braket{i|H|\Phi}|^2+\frac{f(0)}{2}\sum_{i,j=1}^k\frac{(\lambda_i-\lambda_j)^2}{\lambda_jf(\lambda_i/\lambda_j)}|\braket{i|H|j}|^2\\
    &\geq f(0)\sum_{i=1}^k\frac{(\lambda_i-p)^2}{pf(\lambda_i/p)}|\braket{i|H|\Phi}|^2.
\end{align}
From $\lambda_i\leq \delta$ and $1-\delta\leq p$, we have $p-\lambda_i\geq 1-2\delta$. Since $\delta<1/2$ implies $1-2\delta>0$, it holds $(p-\lambda_i)^2\geq (1-2\delta)^2$. On the other hand since $f$ is an operator monotone and hence a monotonic function, we get
\begin{align}
    f\left(\frac{\lambda_i}{p}\right)\leq f\left(\frac{\delta}{1-\delta}\right),
\end{align}
where we used $\lambda_i/p\leq \delta/(1-\delta)$. 
From these inequalities and $p\leq1$, we obtain a lower bound of metric adjusted skew informations as
\begin{align}
     I^f(\rho,H)&\geq \frac{f(0)}{f\left(\frac{\delta}{1-\delta}\right)}(1-2\delta)^2\sum_{i=1}^k|\braket{i|H|\Phi}|^2\\
    &=\frac{f(0)}{f\left(\frac{\delta}{1-\delta}\right)}(1-2\delta)^2\mathrm{Var}(\Phi,H)
\end{align}
\end{proof}

Consider a case where a state $\rho$ is close to a pure state $\Psi$. Equation~\eqref{eq:masi_lowerbound_infidelity} shows that metric adjusted skew informations $I^f(\rho,H)$ of are lower bounded by the variance of the eigenstate $\Phi$ of $\rho$ with largest eigenvalue. Furthermore, Eq.~\eqref{eq:infidelity_largest_ev} shows that $\Phi$ is also close to $\Psi$ in the trace distance. Therefore, a lower bound of $I^f(\rho_m,H_{\ceil{Rm}})$ in Lemma~\ref{lem:local_minima_MASI} can be derived by analyzing the asymptotic behavior of the energy variances of pure states that are close to $\psi^{\otimes \ceil{Rm}}$. 

For this purpose, let us briefly review the asymptotic behavior of the energy distribution of $\psi^{\otimes \ceil{Rm}}$. Let us define the energy distribution
\begin{align}
    p_{\psi^{\otimes \ceil{Rm}}}(E)\coloneqq \Braket{\psi^{\otimes \ceil{Rm}}|\Pi_{E}^m|\psi^{\otimes \ceil{Rm}}},\quad E\in \mathrm{Spec}(H_{\mathrm{iid},\ceil{Rm}}),
\end{align}
where $\mathrm{Spec}(A)$ denotes the set of all eigenvalues of an operator $A$ and $\Pi_E^m$ denotes the projectors to the eigenspace of $H_{\mathrm{iid},\ceil{Rm}}$ with eigenvalue $E$. Since $\psi$ has a period $2\pi$, $p_{\psi^{\otimes \ceil{Rm}}}(E)=0$ if $E\notin \mathbb{Z}$. In the asymptotic limit of $m\to\infty$, it is known that the energy distribution $p_m\coloneqq \{p_{\psi^{\otimes \ceil{Rm}}}(E)\}_{E}$ converges to the Poisson distribution up to a shift. Precisely, under the assumption that the absolute third moment is finite, i.e., $\braket{\psi||H|^3|\psi}<\infty$, it is shown \cite{barbour_total_2002,marvian_operational_2022,yamaguchi_beyond_2022} that there exists $k_m\in\mathbb{Z}$ for each $m$ such that
\begin{align}
    \lim_{m\to\infty}d_{\mathrm{TV}}\left(p_{\psi^{\otimes \ceil{Rm}}},\Upsilon_{k_m}\mathrm{P}_{mR\mathrm{Var}(\psi,H)}\right)=0,\label{eq:convergence_TP}
\end{align}
where $d_{\mathrm{TV}}$ denotes the total variation distance and $\Upsilon_k$ for $k\in\mathbb{Z}$ denotes the shift operation for a probability distribution defined by $(\Upsilon_k p)(n)\coloneqq p(n-k)$.

Now, let us prove an asymptotic behavior of the variance of random variables which approximately follows the Poisson distribution:
\begin{lemma}\label{lem:variance_lowerbound}
Fix a positive parameter $\lambda>0$. Let $q=\{q_m\}_m$ be a sequence of probability distributions such that 
\begin{align}
    \exists M>0,\quad \forall m>M, \quad \exists k_m\in\mathbb{Z},\quad  d_{\mathrm{TV}}\left(q_m,\Upsilon_{k_m}\mathrm{P}_{m\lambda}\right)\leq \epsilon 
\end{align}
for a sufficiently small parameter $\epsilon$. 
Then there exists a function $\gamma_\lambda(\epsilon)$ such that
\begin{align}
    \lim_{\epsilon\to 0}\gamma_\lambda(\epsilon)=0
\end{align}
and
\begin{align}
    \mathrm{Var}\left(q_m\right)
    &\geq \mathrm{Var}\left(\mathrm{P}_{m\lambda}\right)-\gamma_\lambda(\epsilon) m
\end{align}
holds for all sufficiently large $m$.

\end{lemma}

\begin{proof}
For simplicity, we hereafter assume that the probability distributions $q_m$ are defined on integers \footnote{The proof is valid for a general case if we replace the equalities in Eqs.~\eqref{eq:dtv_bound} and \eqref{eq:variance_firstline} with $\leq$ and $\geq$, respectively.}. Since the variance is invariant under the translation operation $\Upsilon_k$, we can assume that $k_m=0$ without loss of generality. 
Let us introduce the following notations: 
\begin{align}
    \delta_{m}(n)&\coloneqq \mathrm{P}_{m\lambda}(n)-q_m(n),\\
    \mu_{\mathrm{P}_{m\lambda}}&\coloneqq \sum_{n}n\mathrm{P}_{m\lambda}(n)=m\lambda\\
    \mu_{q_m}&\coloneqq \sum_{n} nq_m(n),\\
    \Delta\mu_m&\coloneqq  \mu_{\mathrm{P}_{m\lambda}}-\mu_{q_m}=\sum_{n}n\delta_m(n).
\end{align}
From $\mathrm{P}_{m\lambda}(n)\geq 0$ and $q_m(n)\geq 0$, we have 
\begin{align}
    \mathrm{P}_{m\lambda}(n)\geq \delta_m(n)\geq -q_m(n).
\end{align}
Defining $ A_m\coloneqq \left\{n\in\mathbb{Z}\middle| \delta_m(n)>0\right\}$, we have
\begin{align}
    \sum_{n\in A_m}\delta_m(n)=d_{\mathrm{TV}}\left(q_m,\mathrm{P}_{m\lambda}\right).\label{eq:dtv_bound}
\end{align}

The variance of $q_m$ is decomposed into
\begin{align}
    \mathrm{Var}\left(q_m\right)&=\sum_n (n-\mu_{q_m})^2q_m(n)\label{eq:variance_firstline}\\
    &=\sum_n (n-\mu_{q_m})^2\left(\mathrm{P}_{m\lambda}(n)-\delta_m(n)\right)\\
    &=\mathrm{Var}(\mathrm{P}_{m\lambda})+\Delta\mu_{m}^2-\sum_n (n-\mu_{q_m})^2\delta_m(n).
\end{align}
To provide a lower bound of this quantity, let us introduce an interval $I_m\coloneqq [\mu_{\mathrm{P}_{m\lambda}}-a_m,\mu_{\mathrm{P}_{m\lambda}}+a_m]\cap\mathbb{Z}$, where $a_m=\alpha_\epsilon \sqrt{m}$. Here, $\alpha_\epsilon>0$ for $\epsilon>0$ is defined by $\alpha_\epsilon\sqrt{\lambda}\coloneqq g^{-1}\left(1-\epsilon\right)$, where 
\begin{align}
    g(x)\coloneqq\int_{-x}^x \dd\beta \frac{\beta^2}{\sqrt{2\pi}}e^{-\frac{\beta^2}{2}}= \mathrm{Erf}\left(\frac{x}{\sqrt{2}}\right)-\frac{2x}{\sqrt{2\pi}}e^{-\frac{x^2}{2}}.
\end{align}
Note that $\alpha_\epsilon$ is uniquely determined since $g(x)$ is monotonic and takes any values in $[0,1]$ for $x\geq0$. 

Defining $\bar{I}_m\coloneqq\mathbb{Z}\setminus I_m$, we have 
\begin{align}
    &\sum_n (n-\mu_{q_m})^2\delta_m(n)\\
    &=\sum_{n\in I_m} (n-\mu_{q_m})^2\delta_m(n)+\sum_{n\in \bar{I}_m} (n-\mu_{q_m})^2\delta_m(n)\\
    &\leq \sum_{n\in I_m \cap A_m} (n-\mu_{q_m})^2\delta_m(n)+\sum_{n\in \bar{I}_m} (n-\mu_{q_m})^2\mathrm{P}_{m\lambda}(n)\\
    &\leq \left(|\Delta\mu_m|+a_m\right)^2\epsilon+\sum_{n\in \bar{I}_m} (n-\mu_{q_m})^2\mathrm{P}_{m\lambda}(n)\\
    &= \left(|\Delta\mu_m|+a_m\right)^2\epsilon+\Delta\mu_m^2\sum_{n\in\bar{I}_m}\mathrm{P}_{m\lambda}(n)+2\Delta\mu_m\sum_{n\in\bar{I}_m}(n-\mu_{\mathrm{P}_{m\lambda}})\mathrm{P}_{m\lambda}(n)+\sum_{n\in\bar{I}_m}(n-\mu_{\mathrm{P}_{m\lambda}})^2\mathrm{P}_{m\lambda}(n).
\end{align}

By using the Stirling formula
\begin{align}
    n!\sim \sqrt{2\pi n}\left(\frac{n}{e}\right)^n,
\end{align}
it holds
\begin{align}
    \mathrm{P}_{m\lambda}(m\lambda+k)&=\frac{(m\lambda)^{m\lambda+k}}{(m\lambda+k)!}e^{-m\lambda}\\
    &\sim \frac{1}{\sqrt{2\pi (m\lambda+k)}}\left(\frac{m\lambda}{m\lambda+k}\right)^{m\lambda+k}e^{k}\\
    &\sim  \frac{1}{\sqrt{2\pi m\lambda}} e^{-\frac{\beta^2}{2\lambda} }
\end{align}
for $k=\beta \sqrt{m}$ as $m\to\infty$. 

With this formula, we have
\begin{align}
    \lim_{m\to\infty}\frac{1}{\sqrt{m}}\sum_{n\in I_m}(n-\mu_{\mathrm{P}_{m\lambda}})\mathrm{P}_{m\lambda}(n)=\int_{-\alpha_\epsilon}^{\alpha_\epsilon}\dd \beta\frac{\beta}{\sqrt{2\pi\lambda}}e^{-\frac{\beta^2}{2\lambda}}=0,
\end{align}
and
\begin{align}
    \lim_{m\to\infty}\frac{1}{m}\sum_{n\in I_m}(n-\mu_{\mathrm{P}_{m\lambda}})^2\mathrm{P}_{m\lambda}(n)=\int_{-\alpha_\epsilon}^{\alpha_\epsilon}\dd \beta\frac{\beta^2}{\sqrt{2\pi\lambda}}e^{-\frac{\beta^2}{2\lambda}}=\lambda g(\alpha_\epsilon\sqrt{\lambda})=\lambda(1-\epsilon).
\end{align}
implying that 
\begin{align}
   \sum_{n\in\bar{I}_m} (n-\mu_{\mathrm{P}_{m\lambda}})\mathrm{P}_{m\lambda}(n)&=o(\sqrt{m})\\
   \sum_{n\in\bar{I}_m} (n-\mu_{\mathrm{P}_{m\lambda}})^2\mathrm{P}_{m\lambda}(n)&= \lambda\epsilon  m+o(m).
\end{align}
Therefore, for any positive constants $u$ and $v$ which are independent of $m$, 
\begin{align}
    \sum_{n\in\bar{I}_m} (n-\mu_{\mathrm{P}_{m\lambda}})\mathrm{P}_{m\lambda}(n)&\leq u\sqrt{m}\\
    \sum_{n\in\bar{I}_m} (n-\mu_{\mathrm{P}_{m\lambda}})^2\mathrm{P}_{m\lambda}(n)&= (\lambda\epsilon +v ) m
\end{align}
hold for all sufficiently large $m$. For future convenience, we take $u=\epsilon \alpha_\epsilon\sqrt{\lambda}$ and $v=\lambda\epsilon$, which implies that
\begin{align}
   \sum_{n\in\bar{I}_m} (n-\mu_{\mathrm{P}_{m\lambda}})\mathrm{P}_{m\lambda}(n)&\leq \epsilon\alpha_\epsilon\sqrt{m\lambda}\\
   \sum_{n\in\bar{I}_m} (n-\mu_{\mathrm{P}_{m\lambda}})^2\mathrm{P}_{m\lambda}(n)&\leq  2\lambda\epsilon m
\end{align}
hold for all sufficiently large $m$. 

On the other hand, by using Chebyshev's Inequality we have
\begin{align}
    \sum_{n\in \bar{I}_m}\mathrm{P}_{m\lambda}(n)\leq\frac{\lambda}{\alpha_\epsilon^2}.
\end{align}

By using these results, the variance of $q_m$ is lower bounded as
\begin{align}
    \mathrm{Var}\left(q_m\right)
    &=\mathrm{Var}(p_m)+\Delta\mu_{m}^2-\sum_n (n-\mu_{q_m})^2\delta_m(n)\\
    &\geq \mathrm{Var}(p_m)+\Delta\mu_{m}^2-\left( \left(|\Delta\mu_m|+\alpha_\epsilon\sqrt{m}\right)^2\epsilon+\Delta\mu_m^2\frac{\lambda}{\alpha_\epsilon^2}+2\Delta\mu_m\epsilon\alpha_\epsilon\sqrt{m\lambda}+2\lambda\epsilon m\right).
\end{align}
Since the right hand side is quadratic in $\Delta\mu_m$, we get
\begin{align}
    &\Delta\mu_{m}^2-\left( \left(|\Delta\mu_m|+\alpha_\epsilon\sqrt{m}\right)^2\epsilon+\Delta\mu_m^2\frac{\lambda}{\alpha_\epsilon^2}+2\Delta\mu_m\epsilon\alpha_\epsilon\sqrt{m\lambda}+2\lambda\epsilon m\right)\\
    &\geq \left(1-\epsilon-\frac{\lambda}{\alpha_\epsilon^2}\right)|\Delta\mu_{m}|^2-4\alpha_\epsilon\epsilon\sqrt{m\lambda}|\Delta\mu_{m}|-\left(\alpha_\epsilon^2\epsilon+2\lambda\epsilon \right)m \\
    &\geq -\left(\left(\alpha_\epsilon^2\epsilon+2\lambda\epsilon \right)+\frac{4\alpha_\epsilon^2\epsilon^2\lambda}{\left(1-\epsilon-\frac{\lambda}{\alpha_\epsilon^2}\right)}\right)m
\end{align}
where we have assumed that 
\begin{align}
    1-\epsilon-\frac{\lambda}{\alpha_\epsilon^2}>0
\end{align}
in the last line.
This assumption is true for a sufficiently small $\epsilon$ since $\lim_{\epsilon\to0}\alpha_\epsilon
=\lim_{\epsilon\to0}g^{-1}(1-\epsilon)/\sqrt{\lambda}=\infty$ holds.

We now introduce
\begin{align}
    \gamma_\lambda(\epsilon)\coloneqq \left(\alpha_\epsilon^2\epsilon+2\lambda\epsilon \right)+\frac{4\alpha_\epsilon^2\epsilon^2\lambda}{\left(1-\epsilon-\frac{\lambda}{\alpha_\epsilon^2}\right)}.
\end{align}
From Fig.~\ref{fig:alpha_epsilon}, it holds
\begin{align}
    \lim_{\epsilon\to 0} \epsilon \alpha_\epsilon&=\lim_{\epsilon \to 0}\epsilon g^{-1}(1-\epsilon)/\sqrt{\lambda}=0\\
    \lim_{\epsilon\to 0} \epsilon \alpha_\epsilon^2&=\lim_{\epsilon \to 0}\epsilon \left(g^{-1}(1-\epsilon)\right)^2/\lambda=0.
\end{align}
Therefore, we get $ \lim_{\epsilon\to0}\gamma_\lambda(\epsilon)=0$, completing the proof. 
\begin{figure}[htbp]
    \centering
    \includegraphics[width=10cm]{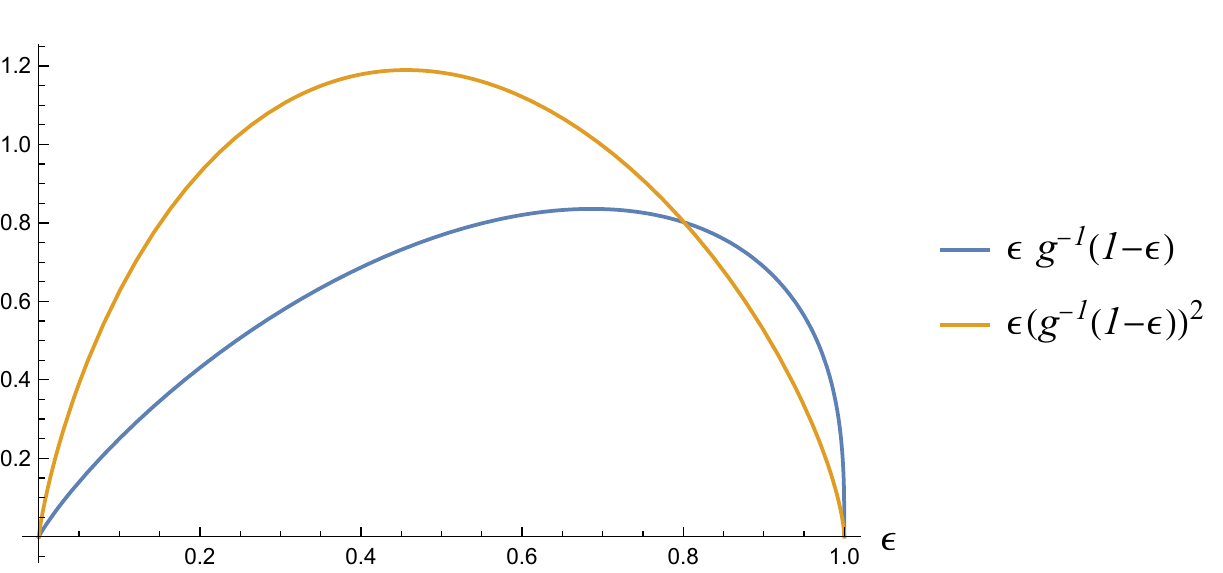}
    \caption{The behaviors of $\epsilon g^{-1}(1-\epsilon)$ and 
    $\epsilon (g^{-1}(1-\epsilon))^2
    $ for $\epsilon\in[0,1]$.}
    \label{fig:alpha_epsilon}
\end{figure}
\end{proof}

Combining these lemmas, we here prove Lemma~\ref{lem:local_minima_MASI}:
\begin{proof}[Proof of Lemma~\ref{lem:local_minima_MASI}]
We denote $\psi_m\coloneqq \psi^{\otimes \ceil{Rm}}$.
Let $\widehat{\rho}=\{\rho_m\}_m$ be an arbitrary sequence of states such that $\rho_m\in B^\epsilon(\psi_m)$. 
From the Fuchs-van de Graaf inequalities, an upper bound of the infidelity of $\rho_m$ and $\psi_m$ is derived as
\begin{align}
    1-\Braket{\psi_m|\rho_m|\psi_m}\leq 1-(1-D(\rho_m ,\psi_m))\leq 1-(1-\epsilon)^2\eqqcolon \delta_1. 
\end{align}

Let $\Phi_m$ be the eigenstate of $\rho_m$ with the largest eigenvalue. From Eq.~\eqref{eq:infidelity_largest_ev}, we get
\begin{align}
    \left|\Braket{\Phi_m|\psi_m}\right|\geq 1-\delta_1. 
\end{align}
For the probability distribution $p_{\Phi_m}$ of $\Phi_m$, it holds
\begin{align}
    d_{\mathrm{TV}}(p_{\Phi_m},p_{\psi_m})\leq D\left(\Phi_m,\psi_m\right)\leq \sqrt{1-(1-\delta_1)^2}.
\end{align}
By using Eq.~\eqref{eq:convergence_TP}, for any $\epsilon>0$,
\begin{align}
    \exists k_m\in\mathbb{Z},\quad 
    d_{\mathrm{TV}}(p_{\psi_m},\Upsilon_{k_m}\mathrm{P}_{m\lambda})\leq\epsilon
\end{align}
holds for all sufficiently large $m$, where $\lambda\coloneqq R\mathrm{Var}(\psi,H)$.
Therefore, from the triangle inequality, we get
\begin{align}
    \exists k_m\in\mathbb{Z},\quad 
    d_{\mathrm{TV}}(p_{\Phi_m},\Upsilon_{k_m}\mathrm{P}_{m\lambda})\leq \sqrt{1-(1-\delta_1)^2}+\epsilon\eqqcolon \delta_2
\end{align}
for all sufficiently large $m$. 
Note that $\lim_{\epsilon\to 0}\delta_1=\lim_{\epsilon\to 0}\delta_2=0$ holds. 
From Lemmas~\ref{lem:masi_lowerbound} and \ref{lem:variance_lowerbound}, we get
\begin{align}
     I(\rho_m,H_{\mathrm{iid},\ceil{mR}})
     &\geq \frac{f(0)}{f\left(\frac{\delta}{1-\delta_1}\right)}(1-2\delta_1)^2\mathrm{Var}(\Phi_m,H_m)\\
     &\geq m\left(\lambda -\gamma_\lambda (\delta_2)\right) \frac{f(0)}{f\left(\frac{\delta_1}{1-\delta_1}\right)}(1-2\delta_1)^2
\end{align}
as $m\to\infty$, where $\gamma_\lambda$ is defined in Lemma~\ref{lem:variance_lowerbound}. 
Defining
\begin{align}
    \delta^f(\epsilon)\coloneqq \left(\lambda -\gamma_\lambda (\delta_2)\right) \frac{f(0)}{f\left(\frac{\delta_1}{1-\delta_1}\right)}(1-2\delta_1)^2-\lambda,
\end{align}
we have $\lim_{\epsilon\to 0}\delta^f(\epsilon)=0$ and
\begin{align}
    I(\rho_m,H_{\mathrm{iid},\ceil{mR}})\geq m\lambda -m\delta^f(\epsilon)=mR\mathrm{Var}(\psi,H)-m\delta^f(\epsilon)
\end{align}
for all sufficiently large $m$. Since 
\begin{align}
    \lim_{m\to\infty}\frac{1}{m}I\left(\psi^{\otimes \ceil{Rm}},H_{\mathrm{iid},\ceil{mR}}\right)=R\mathrm{Var}(\psi,H),
\end{align}
this result can also be written as
\begin{align}
    I(\rho_m,H_{\mathrm{iid},\ceil{mR}})\geq I\left(\psi^{\otimes \ceil{Rm}},H_{\mathrm{iid},\ceil{mR}}\right) -m\delta^f(\epsilon)+o(m)\quad\quad (m\to\infty).
\end{align}

\end{proof}

As an immediate corollary, we get Theorem~\ref{thm:smooth_masi_rates_iid}:
\begin{theorem*}[Restatement of Theorem~\ref{thm:smooth_masi_rates_iid}]
Let $\psi$ be a pure state having period $2\pi$ for a Hamiltonian $H$. Assume that the third absolute moment of the Hamiltonian is finite, i.e., $\braket{\psi||H|^3|\psi}<\infty$. For a positive parameter $R>0$, define $\widehat{\psi}_{\mathrm{iid}}(R)\coloneqq \{\psi^{\otimes \ceil{Rm}}\}_m$ and $\widehat{H}_{\mathrm{iid}}(R)\coloneqq \{H_{\mathrm{iid},\ceil{Rm}}\}_m$, where $H_{\mathrm{iid},k}\coloneqq \sum_{i=1}^k \mathbb{I}^{\otimes i-1}\otimes H\otimes\mathbb{I}^{\otimes k-i} $. The smooth metric adjusted skew information rates for this i.i.d. sequence are given by
\begin{align}
    &I_{+}^f(\widehat{\psi}_{\mathrm{iid}}(R),\widehat{H}_{\mathrm{iid}}(R))= I_{-}^f(\widehat{\psi}_{\mathrm{iid}}(R),\widehat{H}_{\mathrm{iid}}(R))\nonumber\\
    &=\lim_{m\to\infty}\frac{1}{m}I^{f}(\psi^{\otimes \ceil{Rm}},H_{\ceil{Rm}})=I^f\left(\psi,H\right)R.
\end{align}
\end{theorem*}

\begin{proof}
From Lemma~\ref{lem:local_minima_MASI}, we get
\begin{align}
    I_+^f (\widehat{\psi}_{\mathrm{iid}}(R),\widehat{H}_{\mathrm{iid}}(R))\geq I_-^f (\widehat{\psi}_{\mathrm{iid}}(R),\widehat{H}_{\mathrm{iid}}(R))\geq I(\psi,H)R.
\end{align}
A straightforward calculation shows
\begin{align}
    \lim_{m\to\infty}\frac{1}{m}I^f(\psi^{\otimes \ceil{Rm}},H_{\mathrm{iid},\ceil{Rm}})=\lim_{m\to\infty}\frac{1}{m}\ceil{Rm}I(\psi,H)=I(\psi,H)R.
\end{align}
Since  
\begin{align}
    \lim_{m\to\infty}\frac{1}{m}I^f(\psi^{\otimes \ceil{Rm}},H_{\mathrm{iid},\ceil{Rm}})\geq I_+^f (\widehat{\psi}_{\mathrm{iid}}(R),\widehat{H}_{\mathrm{iid}}(R))\geq I_-^f (\widehat{\psi}_{\mathrm{iid}}(R),\widehat{H}_{\mathrm{iid}}(R))
\end{align}
holds by definition of $I^f_\pm$, we get
\begin{align}
    I_+^f (\widehat{\psi}_{\mathrm{iid}}(R),\widehat{H}_{\mathrm{iid}}(R))=I_-^f (\widehat{\psi}_{\mathrm{iid}}(R),\widehat{H}_{\mathrm{iid}}(R))= I(\psi,H)R. 
\end{align}
\end{proof}

As another corollary, we obtain an alternative proof of the converse part of the conversion theory for i.i.d. pure states \cite{marvian_operational_2022}:
\begin{corollary}
Let $\psi$ and $\phi$ be pure states with period $2\pi$ with Hamiltonians $H$ and $H'$, respectively. If $(\widehat{\psi},\widehat{H}_{\mathrm{iid}})\cova (\widehat{\phi}(R),\widehat{H}'_{\mathrm{iid}}(R))$ holds for $\widehat{\psi}=\{\psi^{\otimes m}\}_m$ and $\widehat{\phi}(R)=\{\phi^{\ceil{Rm}}\}_m$, then it holds
\begin{align}
    \mathrm{Var}(\psi,H)\geq \mathrm{Var}(\phi,H') R.
\end{align}
i.e., 
\begin{align}
    R\leq \frac{\mathrm{Var}(\psi,H)}{\mathrm{Var}(\phi,H')}=\frac{\mathcal{F}(\psi,H)}{\mathcal{F}(\phi,H')}.
\end{align}
\end{corollary}
\begin{proof}
From the monotonicity of the smooth metric adjusted skew information rates, we get
\begin{align}
    I^{f}_\pm(\widehat{\psi},\widehat{H})\geq I^{f}_\pm(\widehat{\phi}(R),\widehat{H}'(R)).
\end{align}
From Theorem~\ref{thm:smooth_masi_rates_iid}, this inequality implies
\begin{align}
    I(\psi,H)\geq I(\phi,H')R.
\end{align}
\end{proof}

\section{Review of the spectral quantum Fisher information rates}\label{app:spectral_qfi}

We here briefly review the results in \cite{yamaguchi_beyond_2022}. For a pure state $\psi$ and a Hamiltonian $H$, the period is defined as
\begin{align}
    \tau\coloneqq \inf_{t>0}\left\{t\,\middle|\,|\braket{\psi|e^{-\ii t H}|\psi}|=1\right\}.
\end{align}
We assume that $H$ is bounded below and $0<\tau<\infty$. Without loss of generality, we can set the period to be $2\pi$ by rescaling the Hamiltonian as $H\to\frac{\tau}{2\pi}H$. In the following, we always assume that pure states have period $2\pi$ for simplicity. In this case, the pure state $\psi$ has support in the eigenspaces of the Hamiltonian with eigenvalues given by $n+E_0$, where $n$ is a positive integer and $E_0$ is a constant. Shifting the Hamiltonian by a constant, we can assume that $E_0=0$. We define the energy distribution of a pure state $\psi$ by
\begin{align}
    p_\psi(n)\coloneqq \braket{\psi|\Pi_n|\psi}\quad (n\in \mathbb{Z}_{\geq 0}),
\end{align}
where $\Pi_n$ is a projector to the eigenspace of the Hamiltonian with eigenvalue $n$. An essential fact is that the exact convertibility among pure states is fully characterized by the energy distribution \cite{gour_resource_2008}. This is because $(\psi,H)\cov(\psi_{\mathrm{HO}}',H_{\mathrm{HO}})$ and $(\psi_{\mathrm{HO}}',H_{\mathrm{HO}})\cov (\psi,H)$ holds for a state
\begin{align}
    \ket{\psi_{\mathrm{HO}}'}\coloneqq \sum_{n= 0}^\infty \sqrt{p_{\psi}(n)}\ket{n}\label{eq:ref_state_ho}
\end{align}
of a harmonic oscillator system with Hamiltonian $H_{\mathrm{HO}}=\sum_{n=0}^\infty n\ket{n}\bra{n}$. This observation can also be extended to the convertibility with vanishing error in the asymptotic regime. For further detail, see also \cite{yamaguchi_beyond_2022}.

To define the max- and min-quantum Fisher information, we introduce several notations. For a real number $\lambda\in\mathbb{R}$, we define a generalized Poission distribution $\mathrm{P}_{\lambda}=\{\mathrm{P}_{\lambda}(n)\}_{n\in\mathbb{Z}}$ by
\begin{align}
    \mathrm{P}_{\lambda}(n)\coloneqq 
    \begin{cases}
    e^{-\lambda}\frac{\lambda^n}{n!}&\quad (n\geq 0)\\
    0&\quad (n< 0)
    \end{cases}.
\end{align}
For $\lambda\geq 0$, this is an ordinary Poisson distribution. However, for $\lambda< 0$, it is not a probability distribution since some of the element becomes negative. For sequences of numbers $a=\{a(n)\}_{n\in\mathbb{Z}}$ and $b=\{b(n)\}_{n\in\mathbb{Z}}$, we define the convolution sequence $(a*b)=\{(a*b)(n)\}_{n\in\mathbb{Z}}$ by $(a*b)(n)\coloneqq \sum_{k\in\mathbb{Z}}a(n-k)b(k)$. For a given sequence $q$, $\widetilde{q}$ denotes its ``inverse" sequence with respect to $*$ in the sense that $(q*\widetilde{q})(n)=\delta_{n,0}$, where $\delta_{n,m}$ is the Kronecker delta. We denote $a\geq 0$ for $a=\{a(n)\}_{n}$ if and only if $a(n)\geq 0$ for all $n$. The max- and min-quantum Fisher information are defined by \cite{yamaguchi_beyond_2022}
\begin{align}
    \mathcal{F}_{\max}\left(\psi,H\right)&\coloneqq \inf\left\{4\lambda\mid \mathrm{P}_{\lambda}*\widetilde{p_{\psi}}\geq 0\right\},\\
    \mathcal{F}_{\min}\left(\psi,H\right)&\coloneqq \sup\left\{4\lambda\mid p_{\psi} *\mathrm{P}_{-\lambda}\geq 0\right\}.
\end{align}

To define the smooth max- and min-quantum Fisher information, let us first assume that the system of interest is a harmonic oscillator. The max-quantum Fisher information is extended to a mixed state $\rho$ by $\mathcal{F}_{\max}(\rho,H_{\mathrm{HO}})=\inf_{\Phi_\rho}\mathcal{F}_{\max}\left(\Phi_\rho,H_{\mathrm{HO}}+H_A\right)$, where the infimum is taken over the set of all purification $\Phi_\rho$ of $\rho$ and $H_A$ is a Hamiltonian of an ancillary system with integer eigenvalues. The $\epsilon$-smooth max- and min-quantum Fisher information are then defined as
\begin{align}
    \mathcal{F}^\epsilon_{\max} (\psi,H_{\mathrm{HO}})&\coloneqq \inf_{\rho\in B^\epsilon(\psi)}(\rho,H_{\mathrm{HO}}),\quad \mathcal{F}^\epsilon_{\min} (\psi,H_{\mathrm{HO}})\coloneqq \sup_{\phi\in B_{\mathrm{pure}}^\epsilon(\psi)}\mathcal{F}^\epsilon_{\min} (\phi,H_{\mathrm{HO}}),
\end{align}
where $\epsilon$-balls are defined by
$B^\epsilon(\psi)\coloneqq \{\rho\text{: states}\mid D(\rho,\psi)\leq \epsilon\}$ and $B_{\mathrm{pure}}^\epsilon(\psi)\coloneqq \{\phi\text{: pure states}\mid D(\phi,\psi)\leq \epsilon\}$.

For a system with a generic Hamiltonian, the $\epsilon$-smooth max- and min-quantum Fisher information are defined by
\begin{align}
    \mathcal{F}^\epsilon_{\max} (\psi,H)\coloneqq \mathcal{F}^\epsilon_{\max} (\psi_{\mathrm{HO}}',H_{\mathrm{HO}}),\quad \mathcal{F}^\epsilon_{\min} (\psi,H)\coloneqq \mathcal{F}^\epsilon_{\min} (\psi_{\mathrm{HO}}',H_{\mathrm{HO}}),
\end{align}
where $\psi_{\mathrm{HO}}'$ is defined in Eq.~\eqref{eq:ref_state_ho}. 

Let $(\widehat{\psi},\widehat{H})=(\{\psi_m\}_m,\{H_m\}_m)$ be any sequences of pure states and Hamiltonians with period $2\pi$. The spectral sup- and inf- quantum Fisher information rates are defined as \cite{yamaguchi_beyond_2022}
\begin{align}
    \overline{\mathcal{F}}(\widehat{\psi},\widehat{H})&\coloneqq \lim_{\epsilon\to 0}\limsup_{m\to\infty}\frac{1}{m}\mathcal{F}_{\max}^\epsilon\left(\psi_m,H_m\right),\\
    \underline{\mathcal{F}}(\widehat{\psi},\widehat{H})&\coloneqq \lim_{\epsilon\to 0}\liminf_{m\to\infty}\frac{1}{m}\mathcal{F}_{\min}^\epsilon\left(\psi_m,H_m\right).
\end{align}
These quantities are shown to be equal to the coherence cost and the distillable coherence \cite{yamaguchi_beyond_2022}, i.e., 
\begin{align}
    C_{\mathrm{cost}}(\widehat{\psi},\widehat{H})=\overline{\mathcal{F}}(\widehat{\psi},\widehat{H}),\quad C_{\mathrm{dist}}(\widehat{\psi},\widehat{H})=\underline{\mathcal{F}}(\widehat{\psi},\widehat{H}).
\end{align}

\end{document}